\documentclass[11pt, a4paper]{article}

\usepackage{authblk}
\usepackage{amssymb}
\usepackage{amsmath}
\usepackage{graphicx}
\usepackage[textsize=tiny,colorinlistoftodos]{todonotes}

\usepackage{bm}
\usepackage[amsmath,thmmarks]{ntheorem}

\usepackage[ruled]{algorithm2e}

\newtheorem{theorem}{Theorem}[section]
\newtheorem{proposition}[theorem]{Proposition}
\newtheorem{lemma}[theorem]{Lemma}
\newtheorem{corollary}[theorem]{Corollary}

\theorembodyfont{\rmfamily}
\newtheorem{remarks}{Remark}[section]

\newtheorem{definition}{Definition}[section]

\theoremstyle{nonumberplain}
\theoremsymbol{$\square$}
\newtheorem{proof}{Proof}
\newcommand{\znum}{\mathbb{Z}}
\newcommand{\fnum}{\mathbb{F}}

\newcommand{\point}[1]{\bm{#1}}
\newcommand{\p}{\point}

\newcommand{\field}[1]{\mathbb{#1}}
\newcommand{\tr}[1]{{#1}^t}
\newcommand{\fk}{\field{K}}
\newcommand{\ff}{\mathbb{F}}
\newcommand{\term}[1]{\mathcal{#1}}

\newcommand{\kx}{\field{K}[x_1,\ldots,x_n]}
\newcommand{\grobner}{Gr\"{o}bner }

\DeclareMathOperator{\nf}{NormalForm}

\newcommand{\ideal}[1]{\langle #1 \rangle}
\newcommand{\inner}{\ideal}

\DeclareMathOperator{\lt}{lt}
\DeclareMathOperator{\Next}{Next}
\DeclareMathOperator{\pre}{Pre}
\DeclareMathOperator{\fail}{Fail}
\DeclareMathOperator{\Span}{Span}
\DeclareMathOperator{\im}{i}
\DeclareMathOperator{\zero}{Z}
\DeclareMathOperator{\cat}{cat}

\DeclareMathOperator{\sakata}{BMSUpdate}
\DeclareMathOperator{\fglm}{FGLM}
\DeclareMathOperator{\reduce}{Reduce}
\DeclareMathOperator{\IsGB}{IsGB}

\DeclareMathOperator{\BM}{BerlekampMassey}
\DeclareMathOperator{\sqrfree}{Sqrfree}
\DeclareMathOperator{\shaped}{ShapeDet}
\DeclareMathOperator{\shapen}{ShapePro}
\DeclareMathOperator{\BMSbased}{BMSbased}
\DeclareMathOperator{\toplevel}{TopLevel}

\renewcommand\footnotemark{}

\usepackage{mathptmx}
\usepackage{tikz}

\begin{document}

\title{Sparse FGLM algorithms\let\thefootnote\relax\footnote{\emph{Email addresses:} Jean-Charles.Faugere@inria.fr (Jean-Charles Faug\`{e}re), Chenqi.Mou@lip6.fr (Chenqi Mou)}}

\date{}

\author[a]{Jean-Charles Faug\`{e}re}
\author[b,a]{Chenqi Mou}
\affil[a]{PolSys Project, LIP6, INRIA Paris-Rocquencourt--UPMC--CNRS, 4 place Jussieu, 75005 Paris, France}
\affil[b]{LMIB--School of Mathematics and Systems Science, Beihang University, Beijing 100191, China}
\maketitle

\begin{abstract}
Given a  zero-dimensional ideal $I \subset \kx$ of degree $D$, the transformation of the ordering of its \grobner basis from DRL to LEX  is a key step in polynomial system solving and turns out to be the bottleneck of the whole solving process. Thus it is of crucial importance to design efficient algorithms to perform the change of ordering.

The main contributions of this paper are several efficient methods for the change of ordering which take advantage of the sparsity of multiplication matrices in the classical {\sf FGLM} algorithm.  Combing all these methods, we propose a deterministic top-level algorithm that automatically detects which method to use depending on the input. As a by-product, we have a fast implementation that is able to handle ideals of degree over $40000$. Such an implementation outperforms the {\sf Magma} and {\sf Singular} ones, as shown by our experiments.

First for the shape position case, two methods are designed based on the Wiedemann algorithm: the first is probabilistic and its complexity to complete the change of ordering is $O(D(N_1+n\log (D)))$, where $N_1$ is the number of nonzero entries of a multiplication matrix; the other is deterministic and computes the LEX \grobner basis of $\sqrt{I}$ via Chinese Remainder Theorem. Then for the general case, the designed method is characterized by the Berlekamp--Massey--Sakata algorithm from Coding Theory to handle the multi-dimensional linearly recurring relations. Complexity analyses of all proposed methods are also provided.

Furthermore, for generic polynomial systems, we present an explicit formula for the estimation of the sparsity of one main multiplication matrix, and prove its construction is free. With the asymptotic analysis of such sparsity, we are able to show for generic systems the complexity above becomes $O(\sqrt{6/n \pi} D^{2+\frac{n-1}{n}})$.
\end{abstract}


\section{Introduction}

\subsection{Motivation}
\grobner basis is an important tool in computational ideal theory \cite{B85G, CLO1998U, B93G}, especially for polynomial system solving. For a given ideal and term ordering, the \grobner basis of this ideal with respect to (w.r.t.) the term ordering is a set of generators with good properties, such that manipulation of the ideal can be achieved with these generators.

The term ordering plays an important role in the theory of \grobner bases. It is well-known that \grobner bases w.r.t.\ different term orderings are also different and possess different theoretical and computational properties. For example,  \grobner bases w.r.t.\ the lexicographical ordering (LEX) have good algebraic structures and are convenient to use for polynomial system solving, while those w.r.t.\ the degree reverse lexicographical ordering (DRL) are computationally easy to obtain. Therefore, the common strategy to solve a polynomial system is to first compute the \grobner basis of the ideal defined by the system w.r.t.\ DRL, change its ordering to LEX, and perhaps further convert the LEX \grobner basis to triangular sets \cite{L1992S} or Rational Univariate Representation \cite{R1999S}. That is one of the main usages of algorithms for the change of ordering.

However, the computation of \grobner bases greatly enhanced recently \cite{F1999A, F2002A}, the step to change the ordering  of \grobner bases has become the bottleneck of the whole solving process (see Section \ref{sec:exp}). Hence it is of crucial significance to design efficient algorithms for the change of ordering. The purpose of this paper is precisely to provide such efficient algorithms.

Furthermore, some practical problems can be directly modeled as the change of ordering of \grobner bases. For example, the \grobner basis of an ideal derived from the AES-128 cryptosystem w.r.t.\ a certain term ordering (other than LEX) has been obtained \cite{BPW06}, and it may lead to a successful cryptanalysis on this system if one is able to convert its term ordering to LEX. And the decoding of some cyclic codes can also be regarded as a problem of changing the term ordering \cite{LY1997O}.

\subsection{Related works}
Several algorithms for the change of ordering have already existed, for example the {\sf FGLM} algorithm for the zero-dimensional case \cite{FGLM93E} and the \grobner walk for the generic case \cite{CKM97C}.  Similar algorithms have also been proposed to change the orderings of triangular sets \cite{PC06,DJMS08} or  using the {\sf LLL} algorithm \cite{BF03C} in the bivariate case.

Among them, the {\sf FGLM} algorithm, only applicable to the zero-dimensional case, is an efficient one. The number of field operations it needs to complete the change of ordering is $O(nD^3)$, where $n$ is the number of variables, and $D$ is the degree of the given ideal $I\subset \kx$. Its efficiency may be due to the fact that it reduces the problem of change of ordering to linear algebra operations. Such a connection is achieved through the multiplication matrix $T_i~ (i=1, \ldots, n)$ used in this algorithm, which represents the multiplication by  \(x_i\) in the quotient ring $\fk[x_1, \ldots, x_n]/I$ viewed as a vector space. These matrices are sparse, even when the input polynomial system is dense (see Section \ref{sec: random}). And in this paper we take advantage of  this sparsity structure to obtain fast {\sf FGLM} algorithms with good complexity and performances.

\subsection{Our contributions}
We first study the particular but important case when the zero-dimensional ideal $I$ is in shape position. Two methods based on the Wiedemann algorithm are proposed to compute the \grobner bases of $I$ or $\sqrt{I}$ w.r.t.\ LEX. They both make use of the sparsity by constructing the linearly recurring sequence
$$[\inner{\p{r}, T_1^i \p{e}}: i = 0, \ldots, 2D-1],$$
where $\p{r}$ is a vector and $\p{e} = (1, 0, \ldots)^t$ is the vector representing $\p{1}$ in $\kx/I$. It is easy to see that the minimal polynomial \(f_{1}\) in $\fk[x_1]$ of this linearly recurring sequence is indeed a polynomial in the \grobner basis of $I$ w.r.t LEX ($x_1 < \cdots < x_n$) when \(\deg(f_{1})=D\), and it can be computed by applying the Berlekamp--Massey algorithm \cite{W86S}. Furthermore, we show how to recover efficiently the other polynomials in the \grobner basis by solving structured (Hankel) linear systems. Hence, we are able to complete the first method for the change of ordering to LEX for ideals in shape position with complexity $O(D(N_1+n\log (D)))$, where $N_1$ is the number of nonzero entries in $T_1$. When \(n\ll D\) this almost matches the complexity of computing the minimal polynomial.

The other method for the shape position case uses the deterministic Wiedemann algorithm, which can always return the correct univariate polynomial in the \grobner basis w.r.t.\ LEX. Making use of the Chinese Remainder Theorem, this method adapts and extends the previous one to recover the \grobner basis of $\sqrt{I}$, instead of $I$. Thus it is suitable to those problems where the zeros, instead of the multiplicities, are of interest. For ideals in shape position, this deterministic method can always return the \grobner basis of $\sqrt{I}$ w.r.t.\ LEX with the complexity $O(D(N_1 + D (\log(D) \log\log(D) + n)))$.

We also briefly discuss how to apply an incremental variant of the Wiedemann algorithm to compute the univariate polynomial, which is of special importance among all the polynomials in the \grobner basis. Such an variant has a complexity sensitive to the output, namely the degree of the univariate polynomial, and is efficient when this degree is small.

Then for general ideals to which the methods above may be no longer applicable, we follow the idea above by generalizing the linearly recurring sequence to a $n$-dimensional mapping
$$E:(s_1, \ldots, s_n) \longmapsto \ideal{\p{r}, T_1^{s_1}\cdots T_n^{s_n} \p{e}}.$$
The minimal set of generating polynomials (w.r.t.\ a term ordering) for the linearly recurring relation determined by $E$ is essentially the \grobner basis of the ideal defined by $E$, and this polynomial set can be obtained via the Berlekamp--Massey--Sakata ({\sf BMS} for short hereafter) algorithm from Coding Theory \cite{S88F, S90E}. With modifications of the {\sf BMS} algorithm, we design a method to change the ordering in the general case.
The complexity of this algorithm is  $O(nD(N+\hat{N}\bar{N}D))$, where $N$ is the maximal number of nonzero entries in matrices $T_1, \ldots, T_n$, while $\hat{N}$ and $\bar{N}$ are respectively the number of polynomials and the maximal term number of all polynomials in the resulting \grobner basis.

Combing all these methods above, we present a deterministic top-level algorithm, which is able to choose automatically which method to use according to the input. The efficiency of the proposed methods is verified by experiments. The current implementation outperforms those of {\sf FGLM} in {\sf Magma} and {\sf Singular}. Take a randomly generated quadratic polynomial system of $13$ variables for example, it generates an ideal in shape position of degree $8192$. For such an ideal, the change of ordering to LEX can be achieved in $193.5$ seconds: this is $54$ times faster than the corresponding {\sf Magma} function. As shown in Table \ref{tab:time}, zero-dimensional ideals over a prime field of degree greater than $40000$ are now tractable.

% The following is commented for putting on arxiv
%Take a randomly generated polynomial system of $19$ variables and degree $3$ for example, it generates an ideal in shape position of degree $6859$. For such an ideal, the change of ordering to LEX can be achieved in $15.26$ seconds: this is 71 times faster than the corresponding {\sf Magma} function. As shown in Table \ref{tab:time}, zero-dimensional ideals over a prime field of degree greater than $40000$ are now tractable.

Furthermore, the performances of these methods are heavily dependent on the sparsity of the multiplication matrices, especially $T_1$ for the shape position case. In general we assume the multiplication matrices known. However, for generic polynomial systems consisting of $n$-variate polynomials of degree $d$, the sparsity of $T_1$ is investigated, and we are able to give an explicit formula to compute the number of dense columns in $T_1$ and show indeed its construction is free. These results furnish a complete complexity analysis of the proposed method for generic polynomial systems. Then with an asymptotic analysis of the number of dense columns as $d$ tends to $+\infty$, we show the complexity of the first method for the shape position case becomes $O(\sqrt{6/n \pi} D^{2+\frac{n-1}{n}})$ for generic systems. Such simplified complexity is better than that of {\sf FGLM} with smaller constant and exponent.

\subsection{What is new}
To be self-contained, this paper also includes results obtained in \cite{FM11} in a refined way for description. However, several original extensions have also been presented here, making the discussion on this subject more comprehensive: (1) For ideals in shape position, one new algorithm is proposed based on the deterministic Wiedemann algorithm. Compared with the previous probabilistic one, this algorithm becomes deterministic and aims at the \grobner basis of the radical of the input ideal. (2) The multiplication matrices are assumed known in \cite{FM11}, but here for the multiplication matrix $T_1$ which is of special importance, its sparsity, together with the asymptotic behaviors, and construction cost are analyzed for generic polynomial systems. Such a study furnishes a complete understanding for the complexity of the change of ordering for generic systems, with construction of multiplication matrices also considered. (3) The proof of Theorem \ref{THM:LOOP} is further simplified via introduction of known results in the literature.

% arxiv
%(4) New experimental results are updated.

%From our experimental observations for such systems, we find the number of terms of degree $k$ in the canonical basis correspond to the coefficient of $x^k$ in the expansion of $(1+x+\cdots + x^{(d-1)})^n$ in the generic case. Based on these observations, we are able to give a tight estimation of the sparsity of $T_1$. The complexity to compute all the multiplication matrices are also obtained similarly.

\subsection{Paper structure}
The organization of this paper is as follows. Related preparatory algorithms used in this paper, along with some notations, are first reviewed in Section \ref{sec:pre}. Then Section \ref{sec:shape} is devoted to the shape position case, where two methods with their complexity analyses are exploited. The method based on the {\sf BMS} algorithm for the general case is presented in Section \ref{sec:general}. Section \ref{sec:main} combines all the previous methods to a top-level algorithm. The sparsity of $T_1$ is studied in Section \ref{sec: random} and experimental results are provided in Section \ref{sec:exp}. %and this paper concludes with some remarks in Section \ref{sec:remark}

\section{Backgrounds: FGLM and BMS algorithms}\label{sec:pre}

Let $\kx$ be the $n$-variate polynomial ring over a field $\fk$, with variables ordered as $x_1 < \cdots < x_n$. Suppose $G_1$ is the \grobner basis of a 0-dimensional ideal $I\subset \kx$ w.r.t.\ a term ordering $<_1$. Given another term ordering $<_2$, one wants to compute the \grobner basis $G_2$ of $I$ w.r.t.\ it. Denote by $D$ the degree of $I$, that is, the dimension of $\kx/I$ as a vector space. These notations are fixed hereafter in this paper.

\subsection{FGLM algorithm}

The {\sf FGLM} algorithm is one to perform the change of ordering of \grobner bases of 0-dimensional ideals efficiently \cite{FGLM93E}. The reason why it is fast may be due to the idea that it reduces the problem of ordering change to linear algebra operations in the quotient ring $\kx/I$. Such a reduction is realized in the following way.

First one computes the canonical basis of $\kx/\ideal{G_1}$ and orders its elements according to $<_1$. Let $B=[\epsilon_1, \ldots, \epsilon_D]$ be the ordered basis. Then $\epsilon_1$ will always equal $\p{1}$, for $<_1$ is a term ordering. Given a variable $x_i$, for each element $\epsilon_j$ in $B$, one can compute the normal form of $\epsilon_j x_i$ w.r.t.\ $G_1$, denoted by $\nf(\epsilon_j x_i)$. This normal form, viewed as an element of $\kx/\ideal{G_1}$, can be further written as a linear combination of $B$. Writing the coefficients as a column vector, one can construct a $D\times D$ matrix $T_i$ by adjoining all the column vectors for $j=1, \ldots, D$. This matrix is called the \emph{multiplication matrix} of $x_i$. It is not hard to verify that all $T_i$ commute: $T_i T_j = T_j T_i$ for $i, j = 1, \ldots, n$.

Next one handles all the terms in $\kx$ one by one following $<_2$. For each term $\p{x}^{\p{s}}$ with $\p{s} = (s_1, \ldots, s_n)$, its coordinate vector w.r.t.\ $B$ can be computed by
$$\p{v}_{\p{s}}= T_1^{s_1}\cdots T_n^{s_n}\p{e},$$
where $\p{e}=(1, 0, \ldots, 0)^t$ is the coordinate vector of $\p{1}$. Then criteria proposed in {\sf FGLM} guarantee that once a linear dependency of the coordinate vectors of computed terms
\begin{equation}\label{eq: fglm-linear}
  \sum_{\p{s}\in S} c_{\p{s}}\p{v}_{\p{s}}=0
\end{equation}
is found, a polynomial $f\in G_2$ can be directly derived in the following form
\begin{equation}\label{eq: fglm-gb}
f= \p{x}^{\p{l}}+\sum_{\p{s} \in S, \p{s}\ne \p{l}}\frac{c_{\p{s}}}{c_{\p{l}}}\p{x}^{\p{s}},
\end{equation}
where $\p{x}^{\p{l}}$ is the leading term of $f$ w.r.t.\ $<_2$ (denoted by $\lt(f)$) \cite{FGLM93E}.

As can be seen now, all one needs to do to obtain the \grobner basis $G_2$ is computing the coordinate vector of each term one by one, and checking whether a linear dependency of these vectors occurs after a new vector is computed, which can be realized by maintaining an echelon form of the matrix whose columns are coordinate vectors of previously computed terms. These steps are merely matrix manipulations from linear algebra. A trivial upper bound for the number of terms to consider is $D+1$ because of the vector size.

%The multiplication matrices link the change of ordering of \grobner bases and linear algebra. And this paper

\subsection{BMS algorithm}
The {\sf BMS} algorithm from Coding Theory is a decoding algorithm to find the generating set of the error locator ideal in algebraic geometry codes \cite{S88F, S90E, SH2002A}. From a more mathematical point of view, it computes the set of minimal polynomials (w.r.t.\ a term ordering $<$) of a linearly recurring relation generated by a given multi-dimensional array. It is a generalization of the Berlekamp--Massey algorithm, which is applied to Reed--Solomon codes to find the generating error locator polynomial, or mathematically the minimal polynomial of a linearly recurring sequence.

The {\sf BMS} algorithm, without much modification, can also be extended to a more general setting of order domains \cite{CLO1998U, H1998a}. Combining with the Feng--Rao majority voting algorithm \cite{FR1993D}, this algorithm can often decode codes with more with $(d_{\min}-1)/2$ errors if the error locations are general \cite{BO2006C}, where $d_{\min}$ is the minimal distance. Next a concise description of the {\sf BMS} algorithm is given, focusing on its mathematical meanings.

As a vector $\p{u}=(u_1, \ldots, u_n)\in \znum_{\geq 0}^n$ and a term $\p{x}^{\p{u}}=x_1^{u_1}\cdots x_n^{u_n}\in \kx$ are 1--1 corresponding, usually we do not distinguish one from the other. A mapping $E: \znum_{\geq 0}^n \longrightarrow \fk$ is called a \emph{$n$-dimensional array}. In Coding Theory, the array $E$ is usually a syndrome array determined by the error word \cite{SH2002A}. Besides the term ordering, we define the following partial ordering: for two terms $\p{u}=(u_1,
\ldots, u_n)$ and $\p{v}=(v_1, \ldots, v_n)$, we say that  $\p{u}\prec \p{v}$ if $u_i \leq v_i$ for $i=1,\ldots, n$.

\begin{definition}
  Given a polynomial $f=\sum_{\p{s}}f_{\p{s}}
x^{\p{s}}\in \kx$, a $n$-dimensional mapping $E$ is said to satisfy the \emph{$n$-dimensional linearly recurring relation} with \emph{characteristic polynomial} $f$ if
\begin{equation}\label{eqn:charPoly}
\sum_{\p{s}} f_{\p{s}}E_{\p{s}+\p{r}}=0, \quad \forall \p{r}\succ 0.
\end{equation}
\end{definition}

The set of all characteristic polynomials of the $n$-dimensional linearly
recurring relation for the array $E$ forms an ideal, denoted by
$I(E)$. Again in the setting of decoding when $E$ is a syndrome array, this ideal is called the \emph{error locator ideal} for $E$, and its elements are called \emph{error locators}. The definition of $I(E)$ used here in this paper follows \cite{SH2002A}, and one can easily see that this definition is equivalent to that in \cite{CLO1998U} by \cite[Thereom 23]{SH2002A}.

Furthermore, the set of minimal polynomials for $I(E)$ w.r.t.\ $<$, which the
{\sf BMS} algorithm computes, is actually the \grobner basis of $I(E)$ w.r.t.\ $<$ \cite[Lemma 5]{S90E}. The canonical basis of $\kx/I(E)$ is also called the \emph{delta set} of $E$, denoted by $\Delta(E)$. The term ``delta set" comes from the property that if $\p{u}\in \znum_{\geq 0}^n$ is contained in $\Delta(E)$, then $\Delta(E)$ also contains all elements $\p{v} \in \znum_{\geq 0}^n$ such that $\p{v}\prec \p{u}$.

Instead of studying the infinite array $E$ as a whole, the {\sf BMS}
algorithm deals with a truncated subarray of $E$ up to
some term $\p{u}$ according to the given term
ordering $<$. A polynomial $f$ with $\lt(f)=\p{s}$ is said to be
\emph{valid for $E$ up to $\p{u}$} if either $\p{u}\not \succ \p{s}$ or
$$\sum_{\p{t}} f_{\p{t}}E_{\p{t}+\p{r}}=0, \quad \forall \p{r} ~(0\prec \p{r} \leq\p{u}-\p{s}).$$
$E$ may be omitted if no ambiguity occurs. A polynomial set is said to be valid up to $\p{u}$ if each its polynomial is so.

Similarly to {\sf FGLM}, the {\sf BMS} algorithm also handles terms in $\kx$ one by one according to $<$, so that the polynomial set $F$ it maintains is valid up to the new term. Suppose $F$ is valid up to some term $\p{u}$. When the next term of $\p{u}$ w.r.t.\ $<$, denoted by $\Next(\p{u})$, is considered, the {\sf BMS} algorithm will update $F$ so that it keeps valid up to $\Next(\p{u})$. Meanwhile, terms determined by $\Next(\p{u})$ are also tested whether they are members of $\Delta(E)$. Therefore, more and more terms will be verified in $\Delta(E)$ as the {\sf BMS} algorithm proceeds. The set of verified terms in $\Delta(E)$ after the term $\p{u}$ is called the \emph{delta set up to $\p{u}$} and denoted by $\Delta(\p{u})$. Then we have
$$\Delta(\p{1}) \subset \cdots \subset \Delta(\p{u}) \subset \Delta(\Next(\p{u})) \subset \cdots \subset \Delta(E).$$
After a certain number of terms are considered, $F$ and $\Delta(\p{u})$ will grow to the \grobner basis of $I(E)$ and $\Delta(E)$ respectively.

Next only the outlines of the update procedure mentioned above, which is also the main part of the {\sf BMS} algorithm, are presented as Algorithm \ref{alg: bmsUpdate} for convenience of later use. More details will also be provided in Section \ref{sec:general}. One may refer to \cite{SH2002A, CLO1998U} for a detailed description. In Algorithm \ref{alg: bmsUpdate} below, the polynomial set $G$, called the \emph{witness set}, is auxiliary and will not be returned with $F$ in the end of the {\sf BMS} algorithm.

\begin{algorithm}[h]\label{alg:sakata}
    \KwIn{~\\
    \Indp
    $F$, a minimal polynomial set valid up to $\p{u}$;\\
    $G$, a witness set up to $\p{u}$;\\
    $\Next(\p{u})$, a term;\\
    $E$, a $n$-dimensional array up to $\Next(\p{u})$.}
    \KwOut{~\\
    \Indp
    $F^+$, a minimal polynomial set valid up to $\Next(\p{u})$;\\
    $G^+$, a witness set up to $\Next(\p{u})$.
    }
    \BlankLine
    \begin{enumerate}
    \Indm
    \item Test whether every polynomial in $F$ is valid up to $\Next(\p{u})$
    \item Update $G^{+}$ and compute the new delta set up to $\Next(\p{u})$ accordingly
    \item Construct new polynomials in $F^{+}$ such that they are valid up to $\Next(\p{u})$
    \end{enumerate}
\caption{$(F^+, G^+) :=\sakata(F, G, \Next(\p{u}), E)$}\label{alg: bmsUpdate}
\end{algorithm}

\section{Shape position case: probabilistic, deterministic and incremental algorithms}\label{sec:shape}

In this section, the case when the ideal $I$ is in shape position is studied.

\begin{definition}
An ideal $I\subset \kx$ is said to be \emph{in shape position} if its \grobner basis w.r.t LEX is of the following form
\begin{equation}\label{eq: shape}
[f_1(x_1), x_2-f_2(x_1), \ldots, x_n-f_n(x_1)].
\end{equation}
\end{definition}
One may easily see that $I$ here is 0-dimensional and $\deg(f_1)=D$.

Such ideals take a large proportion in all the consistent ideals and have been well studied and applied \cite{B1994T, R1999S}. The special structure of their \grobner bases enables us to design specific and efficient methods to change the term ordering to LEX. In the following, methods designed for different purposes, along with their complexity analyses, are exploited.

Throughout this section, we assume the multiplication matrix $T_1$ is nonsingular. Otherwise, one knows by the Stichelberger's theorem (cf. \cite[Theorem 2.1]{R1999S}) that $x_1 = 0$ will be a root of the univariate polynomial in $I$'s \grobner basis w.r.t.\ LEX, and sometimes the polynomial system can be further simplified.

\subsection{Probabilistic algorithm to compute \grobner basis of the ideal}\label{subsec: shape-normal}

\subsubsection{Algorithm description}
Given a 0-dimensional ideal $I$, if the univariate polynomial $f_1(x_1)$ in its \grobner basis w.r.t.\ LEX is of degree $D$, then we know $I$ is in shape position.

The way to compute such a univariate polynomial is the Wiedemann algorithm. Consider now the following linearly recurring sequence
\begin{equation}\label{eqn:seq}
s=[\inner{\p{r}, T_1^i \p{e}}: i = 0, \ldots, 2D-1],
\end{equation}
where $\p{r}$ is a randomly generated vector in $\fk^{(D\times 1)}$, $T_1$ is the multiplication matrix of $x_1$, $\p{e}$ is the coordinate vector of $\p{1}$ w.r.t the canonical basis of $\kx/I$, and $\inner{\cdot, \cdot}$ takes the inner product of two vectors. It is not hard to see that the minimal polynomial $\tilde{f}_1$ of the sequence $s$ is a factor of $f_1$. As $D$ is always a bound on the size of the linearly recurring sequence, the Berlekamp--Massey algorithm can be applied to the sequence $s$ to compute $\tilde{f}_1$. Furthermore, if $\deg(\tilde{f}_1)=D$, then $\tilde{f}_1 = f_1$ and $I$ can be verified in shape position.

Suppose $\deg(\tilde{f}_1)=D$ holds and $f_i$ in \eqref{eq: shape} is of the form $f_i = \sum_{k=0}^{D-1} c_{i,k}x_1^k$ for $i=2, \ldots, n$. Then computing the whole \grobner basis of $I$ w.r.t.\ LEX reduces to determining all the unknown coefficients $c_{i,k}$. Before we show how to recover them, some basic results about linearly recurring sequences are recalled.

%\begin{definition}
%\label{def:hankel}
%Let \(T=[t_0,t_1,t_2,\cdots]\) be a sequence of elements in \(\mathbb{K}\) and \(d\) an integer. The \(d\times d\) Hankel matrix is defined as
%$$
%H_{d}(T)=
% \left[ \begin {array}{ccccc}
% { t_0}&{ t_1}&{ t_2}&\cdots&{ t_{d-1}}\\
% { t_1}&{ t_2}&{ t_3}&\cdots&{ t_{d}}\\
% \vdots& \vdots & \vdots & \ddots & \vdots\\
%{ t_{d-1}}&{ t_{d}}&{ t_{d+1}}&\cdots&{ t_{2d-2}}
%\end {array} \right].
%$$
%\end{definition}
%
%\begin{theorem}[{\cite{Jo1989}}]\label{thm:hankel}
%Let \(T=[t_0,t_1,t_2,\cdots]\) be a linearly recurring sequence. Then the minimal polynomial \(M^{(T)}(x)=\sum_{i=0}^{d}m_{i}x^{i}\)of the sequence \(T\) is such that:
%\begin{itemize}
%\item[(i)] \(d={\rm rank}(H_{d}(T))= {\rm rank}(H_{i}(T)) \)  for all \(i> d\);
%%\vspace*{-1mm}
%\item[(ii)] \(\ker(H_{d+1}(T))\) is a vector space of dimension \(1\) generated by \(\tr{(m_0,m_1,\ldots,m_{d})}\).
%\end{itemize}
%\end{theorem}

\begin{definition}
\label{def:hankel}
Let \(s=[s_0,s_1,s_2,\cdots]\) be a sequence of elements in \(\mathbb{K}\) and \(d\) an integer. The \(d\times d\) \emph{Hankel matrix} is defined as
$$
H_{d}(s)=
 \left[ \begin {array}{ccccc}
 { s_0}&{ s_1}&{ s_2}&\cdots&{ s_{d-1}}\\
 { s_1}&{ s_2}&{ s_3}&\cdots&{ s_{d}}\\
 \vdots& \vdots & \vdots & \ddots & \vdots\\
{ s_{d-1}}&{ s_{d}}&{ s_{d+1}}&\cdots&{ s_{2d-2}}
\end {array} \right].
$$
\end{definition}

\begin{theorem}[{\cite{Jo1989}}]\label{thm:hankel}
Let \(s=[s_0,s_1,s_2,\cdots]\) be a linearly recurring sequence. Then the minimal polynomial \(M^{(s)}(x)=\sum_{i=0}^{d}m_{i}x^{i}\)of the sequence \(s\) is such that:

\begin{itemize}
\item[(i)] \(d={\rm rank}(H_{d}(s))= {\rm rank}(H_{i}(s)) \)  for all \(i> d\);
\item[(ii)] \(\ker(H_{d+1}(s))\) is a vector space of dimension \(1\) generated by \(\tr{(m_0,m_1,\ldots,m_{d})}\).
\end{itemize}
\end{theorem}

For each $i=2, \ldots, n$, as $x_i-\sum_{k=0}^{D-1} c_{i,k}x_1^k \in I$, one has $\nf( x_i - \sum_{k=0}^{D-1} c_{i,k}x_1^k)=0$, thus
$$\p{v}_{i}:=T_i\p{e} = \sum_{k=0}^{D-1} c_{i,k}\cdot T_1^k \p{e}.$$
Multiplying $T_1^j$ and taking the inner product with a random vector $\p{r}$ to both hands for $j=1,\ldots, D-1$, one can further construct $D$ linear equations
\begin{equation}\label{eqn:Hankel}
  \ideal{\p{r}, T_1^j \p{v}_i} = \sum_{k=0}^{D-1}c_{i,k} \cdot \ideal{\p{r}, T_1^{k+j} \p{e}}, \quad  j=0, \ldots, D-1.
\end{equation}
With $c_{i,k}$ considered as unknowns, the coefficient matrix $H$ with entries $\ideal{\p{r}, T_1^{k+j} \p{e}}$ is indeed a $D\times D$ Hankel matrix, and thus invertible by Theorem \ref{thm:hankel}. Furthermore, the linear equation set \eqref{eqn:Hankel} with the Hankel matrix $H$ can be efficiently solved \cite{Brent1980}. All the solutions of these linear systems for $i=2, \ldots, n$ will lead to the \grobner basis we want to compute.

The method above is summarized in the following algorithm, whose termination and correctness are direct results based on previous discussions. The subfunction $\BM()$ is the Berlekamp--Massey algorithm, which takes a sequence over $\fk$ as input and returns the minimal polynomial of this sequence \cite{W86S}.

\begin{algorithm}[h]%\linesnumbered
    \KwIn{$G_1$, \grobner basis of a 0-dimensional ideal
    $I \subset \kx$ w.r.t.\ $<_1$}
    \KwOut{$G_2$, \grobner basis of $I$ w.r.t.\ LEX if the polynomial returned by $\BM()$ is of degree $D$;
     {\sf Fail}, otherwise.}

\BlankLine

Compute the canonical basis of $\kx/\ideal{G_1}$ and multiplication matrices $T_1, \ldots, T_n$\;

$\p{e} := (1, 0, \ldots, 0)^{t}\in \fk^{(D \times 1)}$\;

%Compute $T_1, \ldots, T_n$ the multiplication matrices\;

Choose $\p{r}_{0}=\p{r}\in \fk^{(D\times 1)}$ randomly\;
\For{$i=1,\ldots,2D-1$}
{
    \(\p{r}_{i} := (T_1^t) \p{r}_{i-1}\)\;\label{line:matrix-vector}
}

Generate the sequence $s := [\inner{\p{r}_i, \p{e}}:\, i=0, \ldots, 2D-1]$\; \label{line:berlekamp-begin}

$f_1 := \BM(s)$\; \label{line:bm}

\eIf{$\deg(f_1)=D$}
{
    \(H:=H_{D}(s)\)  \hfill// {\it Construct the Hankel matrix}\\
    \For{$i=2, \ldots, n$}
    {
         \(\p{b} :=\tr{\left(\inner{\p{r}_j,T_i \p{e}} :\, j=0,\ldots,D-1\right)}\)\;
        Compute $\p{c}=\tr{(c_{1},\ldots,c_{D})}\ :=H^{-1} \p{b}$\;
        $f_i := \sum_{k=0}^{D-1}\p{c}_{k+1} x_1^k$\;
    }
    \Return $[f_1, x_2-f_2, \ldots, x_n - f_n]$\; \label{line:berlekamp-end}
}
{\Return{\sf Fail}\;}
\caption{Shape position (probabilistic)
 $G_2 := \shapen(G_1, <_1)$}\label{alg: shapeN}
\end{algorithm}

\begin{remarks}\label{rm: normal}
  As can be seen from the description of Algorithm \ref{alg: shapeN}, such a method is a probabilistic one. That is to say, it can return the correct \grobner basis w.r.t.\ LEX with probabilities, and may also fail even when $I$ is indeed in shape position.
\end{remarks}

\subsubsection{Complexity}
In this complexity analysis and others to follow, we assume that the multiplication matrices are all known and neglect their construction cost.

Suppose the number of nonzero entries in $T_1$ is $N_1$. The Wiedemann algorithm (both construction of the linearly recurring sequence and computation of its minimal polynomial with the Berlekamp--Massey algorithm) will take $O(D\,(N_1+\log(D)))$ field operations to return the minimal polynomial $\tilde{f}_1$ \cite{W86S}.

Next we show how the linear system \eqref{eqn:Hankel} can be generated for free. Note that for any $\p{a}, \p{b}\in \fk^{(D\times 1)}$ and $T\in \fk^{(D\times D)}$, we have $\inner{\p{a}, T\p{b}}=\inner{T^t \p{a}, \p{b}}$, where $T^t$ denotes the transpose of $T$. Thus in \eqref{eqn:seq} and \eqref{eqn:Hankel}
\begin{equation*}
    \inner{\p{r}, T_1^i \p{e}} = \inner{(T_1^t)^i \p{r}, \p{e}}, \quad
    \inner{\p{r}, T_1^j \p{v}_i} = \inner{(T_1^t)^j \p{r}, \p{v}_i}.
\end{equation*}
 Therefore, when computing the sequence \eqref{eqn:seq}, we can record $(T_1^t)^i \p{r}~(i=0, \ldots, 2D-1)$ and use them for construction of the linear equation set \eqref{eqn:Hankel}.

First, as each entry $\ideal{\p{r}, T_1^{k+j} \p{e}}$ of the Hankel matrix $H$ can be extracted from the sequence \eqref{eqn:seq}, the construction of $H$ is free of operations. What is left now is the computation of \(\inner{(T_1^t)^j \p{r}, \p{v}_i}\), where \((T_1^t)^j \p{r} \) has already been computed and \(\p{v}_{i}=T_i \p{e}=\nf(x_i)\). Without loss of generality, we can assume that \(\nf(x_i)=x_i\) (this is not true only if there is a linear equation \(x_i+\cdots\) in the Gr\"obner basis \(G_{1}\), and in that case we can eliminate the variable \(x_{i}\)). Consequently \(\p{v}_{i}\) is a vector with all its components equal to $0$ except for one component equal to \(1\). Hence computing \(\inner{(T_1^t)^j \p{r}, \p{v}_i}\) is equivalent to extracting some component from the vector \((T_1^t)^j \p{r}\) and there is not additional cost.

For each $i=2, \ldots, n$, solving the linear equation set $H\p{c}=\p{b}_i$ only needs $O(D\log (D))$ operations
if fast polynomial multiplication is used~\cite{Brent1980}. Summarizing the analyses above, we have the following complexity result for this method.

\begin{theorem}
\label{prop:complexity shape}
Assume that  \(T_1\)\ is constructed (note that \(T_2, \ldots, T_n\) are not needed). If the minimal polynomial of \eqref{eqn:seq} computed by the Berlekamp--Massey algorithm is of degree $D$, then the complexity of this method is bounded by
\begin{equation*}
    O(D(N_1+\log(D))+(n-1)D\log (D))=O(D(N_1+n\log (D))).
    \end{equation*}
\end{theorem}

This complexity almost matches that of computing the minimal polynomial of the multiplication matrix $T_1$ if $n$ is small compared with $D$.

\subsubsection{Example}
We use the following small example to show how this method applies to ideals in shape position. Given the \grobner basis of a 0-dimensional ideal $I\subset \ff_{11}[x_1, x_2, x_3]$ w.r.t.\ DRL
$$G_1 = [x_2^2+9\,x_2+2\,x_1+6, ~x_1^2+2\,x_2+9, ~x_3+9],$$
we first compute the degree of $I$ as $D=4$, the canonical basis $B = [1, x_1, x_2, x_1x_2]$, and the multiplication matrices $T_1$, $T_2$ and $T_3$.

With the random vector $\p{r} = (8,4,8,6)^t \in \fk^{(4\times 1)}$, we can construct the linearly recurring sequence
$$s = [8,4,0,7,6,8,10,10].$$
Then the Berlekamp--Massey algorithm is applied to $s$ to obtain the minimal polynomial $\tilde{f}_1 = x_1^4+8\,x_1+9$. From the equality $\deg(\tilde{f}_1) = D=4$, we know now the input ideal $I$ is in shape position.

The Hankel coefficient matrix
$$H = \left(
\begin{tabular}{cccc}
  8 &4&0&7\\
  4&0&7&6\\
  0&7&6&8\\
  7&6&8&10
\end{tabular}
\right)
$$
is directly derived from $s$. Next take the computation of the polynomial $x_2 - f_2(x_1) \in G_2$ for example, the vector $\p{b}=(8,6,8,3)^t$ is constructed. The solution of the linear equation set $H\p{c} =\p{b}$ being $\p{c} = (1,0,5,0)^t$, we obtain the polynomial in $G_2$ as $x_2 + 6\,x_1^2 + 10$. The other polynomial $x_3 - f_3(x_1)$ can be similarly computed. In the end, we have the \grobner basis of $I$ w.r.t.\ LEX
$$G_2 = [x_1^4+8\,x_1+9, ~x_2 + 6\,x_1^2 + 10, ~x_3 + 9].$$

\subsection{Deterministic algorithm to compute Gr\"obner basis
of radical of the ideal}

As already explained in Remarks \ref{rm: normal}, the classical Wiedemann algorithm is a probabilistic one. For a vector chosen at random, it may only return a proper factor $\tilde{f}_1$ of the polynomial $f_1$, i.e., $\tilde{f}_1 | f_1$ but $\tilde{f}_1 \neq f_1$. In fact, the deterministic Wiedemann algorithm can be applied to obtain the univariate polynomial $f_1$, then one knows for sure whether $I$ is in shape position or not. The main difficulty is to compute the other polynomials \(f_{2},\ldots,f_{n}\) in a deterministic way.

In the following we present an algorithm to compute the Gr\"obner basis of the radical of the ideal \(I\). Indeed, in most applications, only the zeros of a polynomial system are of interest and we do not need to keep  their multiplicities. Hence it is also important to design an efficient method to perform the change of ordering of \grobner basis of an ideal $I$ in a way that the output is the \grobner basis of $\sqrt{I}$.

\subsubsection{Deterministic version of the Wiedemann algorithm}

The way how this deterministic variant of the Wiedemann algorithm
proceeds is first recalled. Instead of a randomly chosen vector in the classical Wiedemann algorithm, in the deterministic version all the vectors of the canonical basis of $\fk^{(D\times 1)}$
$$\p{e}_1 = (1, 0,\ldots, 0)^t, \p{e}_2 = (0,1, 0, \ldots, 0)^t, \ldots, \p{e}_D = (0, \ldots, 0, 1)^t$$
are used. One first computes the minimal polynomial $f_{1,1}$ of the linearly recurring sequence
\begin{equation}\label{eq: deterministic-1}
  [\inner{\p{e}_1, T_1^j \p{e}}: j = 0, \ldots, 2D-1].
\end{equation}
Suppose $d_1 = \deg(f_{1,1})$, and $\p{b}_1 = f_{1,1}(T_1)\p{e}$. If $\p{b}_1 = \p{0}$, one has $f_{1,1}=f_1$ and the algorithm ends; else it is not hard to see that the minimal polynomial $f_{1,2}$ of the sequence
$$[\inner{\p{e}_2, T_1^j \p{b}_1}: j = 0, \ldots, 2(D-d_1)-1]$$
is indeed a factor of $f_1/f_{1,1}$, a polynomial of degree $\leq D-d_1$ (that is why only the first $2(D-d_1)$ terms are enough in the above sequence). Next, one computes $\p{b}_2 = f_{1,1}f_{1,2}(T)\p{e}$ and checks whether $\p{b}_2 = \p{0}$. If not, the above procedure is repeated and so on. This method ends with $r~(\leq D)$ rounds and one finds $f_1 = f_{1, 1} \cdots f_{1,r}$.

\subsubsection{Deterministic algorithm description}

First we study the general case when a factor of $f_1$ is found. Suppose a vector $\p{w} \in \fk^{(D\times 1)}$ is chosen to construct the linearly recurring sequence
\begin{equation}\label{eq: det-seq}
  [\inner{\p{w}, T_1^i \p{e}}: i = 0, \ldots, 2D-1],
\end{equation}
and the minimal polynomial of this sequence is $\tilde{f}_1$, a proper factor of $f_1$ of degree $d$. We show how to recover the \grobner basis of $I+\ideal{\tilde{f}_1}$ w.r.t.\ LEX. Since the ideal \(I\) is in shape position, it is not hard to see that the ideal $I+\ideal{\tilde{f}_1}$ is also in shape position, and its \grobner basis w.r.t.\ LEX is indeed $[\tilde{f}_1, x_2-\tilde{f}_2, \ldots, x_n - \tilde{f}_n]$, where $\tilde{f}_i$ is the remainder of $f_i$ modulo $\tilde{f}_1$ for $i=2, \ldots, n$.

Now for each $i$, we can construct the linear system similar to \eqref{eqn:Hankel}
\begin{equation}\label{eq: det-eqn}
  \ideal{\p{w}, T_1^j T_i \p{e}} = \sum_{k=0}^{d-1}y_k \cdot \ideal{\p{w}, T_1^{k+j} \p{e}}, \quad  j=0, \ldots, d-1,
\end{equation}
where $y_0, \ldots, y_{d-1}$ are the unknowns. As the $d\times d$ Hankel matrix of \eqref{eq: det-seq} is invertible by Theorem \ref{thm:hankel}, there is a unique solution $c_{i, 0}, c_{i,1}, \ldots, c_{i, d-1}$ for \eqref{eq: det-eqn}. Next we will connect this solution and a polynomial in the \grobner basis of $I +\ideal{\tilde{f}_1}$, and the following lemma is useful to show this connection.

\begin{lemma}
\label{lem:share}
Suppose \(\tilde{f}_1\) is the minimal polynomial of \eqref{eq: det-seq}
for some  $\p{w} \in \fk^{(D\times 1)}$,
\(\tilde{T}_1\) the multiplication matrix of $x_1$ of the ideal
\(I+\ideal{\tilde{f}_1}\) w.r.t.\ $<_1$, and $\tilde{\p{e}} = (1, 0, \ldots, 0) \in \fk^{(d\times 1)}$ the canonical basis of $\p{1}$ in $\kx/(I+\ideal{\tilde{f}_1})$. Then \(\tilde{f}_1\) is also the minimal polynomial of $[\tilde{\p{e}}, \tilde{T}_1 \tilde{\p{e}}, \tilde{T}_1^2 \tilde{\p{e}}, \ldots]$.
 \end{lemma}

\begin{proof}
Suppose $\tilde{f}_1 = x_1^d + \sum_{k=0}^{d-1} a_k x_1^k$. Then according to {\sf FGLM} criteria, for the ideal \(I+\ideal{\tilde{f}_1}\),
$$\tilde{T}_1^d \p{e} = \sum_{k=0}^{d-1} a_k \tilde{T}_1^k \p{e}$$
is the first linear dependency of the vectors $\tilde{\p{e}}, \tilde{T}_1 \tilde{\p{e}}, \tilde{T}_1^2 \tilde{\p{e}}, \ldots$ when one checks the vector sequence
\([\tilde{\p{e}}, \tilde{T}_1 \tilde{\p{e}}, \tilde{T}_1^2 \tilde{\p{e}}, \ldots]\). That is to say, \(\tilde{f}_1\) is also the minimal polynomial of \([\tilde{\p{e}}, \tilde{T}_1 \tilde{\p{e}}, \tilde{T}_1^2 \tilde{\p{e}}, \ldots]\).
%On the other hand, from the assumption of the lemma,
%\(\tilde{f}_1\) is also the minimal polynomial of the previous linearly recurring sequence. The two sequences share the same  minimal polynomial they share\(\tilde{f}_1\).
%\todo{JCF: please improve the proof}
\end{proof}

\begin{proposition}\label{prop: det}
  Suppose $\p{w} \in \fk^{(D\times 1)}$ is such a vector that a proper factor $\tilde{f}_1$ of $f_1$ of degree $d<D$ is found from the linearly recurring sequence \eqref{eq: det-seq}. Then for each $i=2, \ldots, n$, the polynomial $x_i - \sum_{k=0}^{d-1} c_{i, k} x_1^k$, where $c_{i, 0}, c_{i,1}, \ldots, c_{i, d-1}$ is the unique solution of \eqref{eq: det-eqn}, is in the \grobner basis of $I+\ideal{\tilde{f_1}}$ w.r.t.\ LEX.
\end{proposition}

\begin{proof}
  Let $\tilde{T}_1, \ldots, \tilde{T}_d$ be the multiplication matrices of the ideal $I+\ideal{\tilde{f}_1}$ w.r.t.\ $<_1$.

For each $i=2, \ldots, n$, suppose $x_i - \sum_{k=0}^{d-1} \tilde{c}_{i, k} x_1^k$ is the corresponding polynomial in the \grobner basis of $I+\ideal{\tilde{f}_1}$ w.r.t.\ LEX. Then $\tilde{T}_i \tilde{\p{e}} = \sum_{k=0}^{d-1} \tilde{c}_{i,k} \tilde{T}_1^k \tilde{\p{e}}$ holds, and for any vector $\tilde{\p{w}} \in \fk^{(d\times 1)}$, we have
  \begin{equation*}
  \ideal{\tilde{\p{w}}, \tilde{T}_1^j \tilde{T}_i \tilde{\p{e}}} = \sum_{k=0}^{d-1}\tilde{c}_{i,k} \cdot \ideal{\tilde{\p{w}}, \tilde{T}_1^{k+j} \tilde{\p{e}}}, \quad  j=0, \ldots, d-1.
\end{equation*}
As long as $\tilde{\p{w}}$ is chosen such that the coefficient matrix is invertible, the coefficients $\tilde{c}_{i, 0}, \tilde{c}_{i,1}, \ldots, \tilde{c}_{i, d-1}$ will be the unique solution of the linear equation set
  \begin{equation}\label{eq: det-eqn2}
  \ideal{\tilde{\p{w}}, \tilde{T}_1^j \tilde{T}_i \tilde{\p{e}}} = \sum_{k=0}^{d-1}y_k \cdot \ideal{\tilde{\p{w}}, \tilde{T}_1^{k+j} \tilde{\p{e}}}, \quad  j=0, \ldots, d-1.
\end{equation}

Therefore, to prove the correctness of the proposition, it suffices to show that there exists $\tilde{\p{w}} \in \fk^{(d\times 1)}$ such that the coefficient matrix of \eqref{eq: det-eqn2} is invertible, and that the two linear equation sets \eqref{eq: det-eqn} and \eqref{eq: det-eqn2} share the same solution. In particular, we will prove \eqref{eq: det-eqn} and \eqref{eq: det-eqn2} are the same themselves for some $\tilde{\p{w}}$.

To prove that, we need to show the two Hankel matrices and the vectors in the left hands of \eqref{eq: det-eqn} and \eqref{eq: det-eqn2} are the same. That is, for some vector $\tilde{\p{w}}$
\begin{enumerate}
  \item[(i)] $\inner{\p{w}, T_1^j \p{e}} = \inner{\tilde{\p{w}}, \tilde{T}_1^j \tilde{\p{e}}}$,  for $j=0, \ldots, 2d-2$; \label{item: det-1}
  \item[(ii)] $\inner{\p{w}, T_1^j T_i \p{e}} = \inner{\tilde{\p{w}}, \tilde{T}_1^j \tilde{T}_i \tilde{\p{e}}}$, for $j=0, \ldots, d-1$. \label{item:det-2}
\end{enumerate}
Next we will prove these two arguments respectively.

(i) We take the first $d$ equations in (i)
\begin{equation*}
  \inner{\p{w}, T_1^j \p{e}} = \inner{\tilde{\p{w}}, \tilde{T}_1^j \tilde{\p{e}}} ,\quad  j=0, \ldots, d-1.
\end{equation*}
As the vectors $\tilde{\p{e}}, \tilde{T}_1 \tilde{\p{e}}, \ldots, \tilde{T}_1^{d-1} \tilde{\p{e}}$ are linearly independent, the above linear equation set has a unique solution
$\overline{\p{w}}$  for the unknown $\tilde{\p{w}}$. From Lemma~\ref{lem:share}, the vector sequence $[\tilde{\p{e}}, \tilde{T}_1 \tilde{\p{e}}, \tilde{T}_1^2 \tilde{\p{e}}, \ldots]$ and the sequence \eqref{eq: det-seq} share the same minimal polynomial $\tilde{f}_1$ of degree $d$. Thus there exist $a_0, \ldots, a_{d-1} \in \fk$ such that
\begin{equation*}
    \tilde{T}_1^d \tilde{\p{e}}= \sum_{k=0}^{d-1} a_k \tilde{T}_1^k \tilde{\p{e}}, \quad \inner{\p{w}, T_1^d \p{e}} = \sum_{k=0}^{d-1} a_k \inner{\p{w}, T_1^k \p{e}}.
\end{equation*}
Hence
\begin{equation*}
    \inner{\overline{\p{w}}, \tilde{T}_1^d\tilde{\p{e}}} =  \inner{\overline{\p{w}}, \sum_{k=0}^{d-1} a_k \tilde{T}_1^k \tilde{\p{e}}}
    = \sum_{k=0}^{d-1}a_k \inner{\overline{\p{w}}, \tilde{T}_1^k \tilde{\p{e}}} = \sum_{k=0}^{d-1} a_k \inner{\p{w}, T_1^k \p{e}}
    = \inner{\p{w}, T_1^d \p{e}}.
\end{equation*}
Other equalities in (i) for $j=d+1, \ldots, 2d-2$ can also be proved similarly. Actually, the equality $\inner{\p{w}, T_1^j \p{e}} = \inner{\overline{\p{w}}_0, \tilde{T}_1^j \tilde{\p{e}}}$ holds for any $j = 0, 1, \ldots$.

(ii) Since there is a polynomial $x_i - \sum_{k=0}^{D-1} a'_k x_1^k$ in the \grobner basis of $I$ w.r.t.\ LEX, where $a'_0, \ldots, a'_{D-1} \in \fk$, we know $T_i \p{e} = \sum_{k=0}^{D-1} a'_k T_1^k \p{e}$. Then on one hand, for the vector $\p{w}$ and any $i=0, \ldots, d-1$, we have
$$\inner{\p{w}, T_1^j T_i \p{e} }= \sum_{k=0}^{D-1} a'_k \inner{\p{w}, T_1^{k+j} \p{e}}.$$
On the other hand, as $x_i - \sum_{k=0}^{D-1} a'_k x_1^k \in I$, we have $x_i - \sum_{k=0}^{D-1} a'_k x_1^k \in I+\ideal{\tilde{f}_1}$, and thus $\tilde{T}_i \tilde{\p{e}} = \sum_{k=0}^{D-1} a'_k \tilde{T}_1^k \tilde{\p{e}}$. Therefore for the vector $\overline{\p{w}}$ and any $j=0, \ldots, d-1$,
$$\inner{\overline{\p{w}}, \tilde{T}_1^j \tilde{T}_i \tilde{\p{e}}} = \sum_{k=0}^{D-1} a'_k \inner{\overline{\p{w}}, \tilde{T}_1^{k+j}  \tilde{\p{e}}} =\sum_{k=0}^{D-1} a'_k \inner{\p{w}, T_1^{k+j}\p{e}} = \inner{\p{w}, T_1^j T_i \p{e} }. $$
This ends the proof.
\end{proof}

%\begin{remarks}
%\todo{JCF:\ I think we can remove this remark.}
%  For our experiments, the vector $\overline{\p{w}} \in \fk^{(d \times 1)}$ is usually truncated from $\p{w}\in \fk^{(D \times 1)} $. That is to say, suppose $\p{w} = (w_1, \ldots, w_d, \ldots, w_D)$, then usually $\overline{\p{w}} = (w_1, \ldots, w_d)$.
%\end{remarks}

Now let us return to the special case of the deterministic Wiedemann algorithm, where unit vectors are used to find $f_1= f_{1,1} \cdots f_{1,r}$ with $r \leq D$ and $\deg(f_{1, i}) = d_i$. Suppose $\deg(f_1) = D$ so the ideal $I$ is verified in shape position. In the $i$th step of the algorithm, the unit vector $\p{e}_i$ is applied to construct the linearly recurring sequence
$$[\inner{\p{e}_i, T_1^j \p{b}_{i-1}}: j = 0, \ldots, 2(D-\prod_{k=1}^{i-1} d_k)-1],$$
where $\p{b}_{i-1} = \prod_{k=1}^{i-1}f_{1, k}(T_1) \p{e}$. With this sequence the factor $f_{1,i}$ is computed. As the above sequence is the same as
$$[\inner{(\prod_{k=1}^{i-1}f_{1, k}(T_1))^t\p{e}_i, T_1^j \p{e}}: j = 0, \ldots, 2(D-\prod_{k=1}^{i-1} d_k)-1],$$
from Proposition \ref{prop: det} we can recover efficiently the \grobner basis of $I+\ideal{f_{1,i}}$ w.r.t.\ LEX by constructing and solving linear equation sets with Hankel coefficient matrices.

So we have at hands the factorization $f_1= f_{1,1} \cdots f_{1,r}$,  together with the \grobner basis of $I+\ideal{f_{1, i}}$ w.r.t.\ LEX for $i=1, \ldots, r$. Suppose the \grobner basis for $i$ is
\begin{equation}\label{eq: pi}
P_i = [f_{1, i}, x_2 - f_{2, i}, \ldots, x_n - f_{n, i}].
\end{equation}
Then to recover the polynomials $f_j$ in \eqref{eq: shape} for $j=2, \ldots, n$, we have the following modulo equation set constructed from $P_1, \ldots, P_r$:
\begin{equation}\label{eq: crt}
   \left\{\begin{split}
        f_j &\equiv f_{j,1} \mod f_{1,1}\\
        &\cdots \\
        f_j &\equiv f_{j,r} \mod f_{1,r}
   \end{split}
   \right..
\end{equation}
Now it is natural to give a try of the Chinese Remainder Theorem (short as {\sf CRT} hereafter).

%\begin{proposition}[CRT\ {\cite[Corollary 5.3, Algorithm 5.4]{VG2003M}}]\label{prop:crt}
%  Let $p_1, \ldots, p_r \in \fk[x]$ be polynomials which are pairwise coprime and $f_1, \ldots, f_r \in \fk[x]$. Then there exists $f \in \fk[x]$ such that $f \equiv f_i \mod p_i$ for $i=1, \ldots, r$, and such $f$ is unique under the condition that $\deg(f) < \sum_{i=1}^r\deg(f_i)$.
%
%Furthermore, let $p=\prod_{i=1}^r p_i$, $q_i=p/p_i$, and $h_i$ is such that $h_i q_i \equiv 1 \mod p_i$. Let $c_i $ be the remainder of $h_i f_i$ modulo $ p_i$. Then the unique solution $f = \sum_{i=1}^r c_i q_i$.
%\end{proposition}

To use the {\sf CRT}, we have to check first whether $f_{1, 1}, \ldots, f_{1,r}$ are pairwise coprime. One simple case is when $f_1$ is squarefree, or in other words the input ideal $I$ is radical itself. In that case, the direct application of {\sf CRT} will lead to the \grobner basis $G$ of $I$ w.r.t.\ LEX, and the change of ordering ends.

When the polynomial $f_1$ is not squarefree, the {\sf CRT} does not apply directly. In this case, the \grobner basis of $\sqrt{I}$ w.r.t.\ LEX is our aim. Before the study on how to recover this \grobner basis, we first make clear how a polynomial set of form \eqref{eq: shape} can be split to a series of polynomial sets with a certain zero relation according to some factorization of $f_1$. The following proposition is a direct result of \cite[Proposition 5(i)]{L1992S}, and it is actually a splitting technique commonly used in the theory of  triangular sets \cite{W2001E}. In what follows, $\zero(F)$ denotes the common zeros of a polynomial set $F\in \kx$ in $\overline{\fk}^n$, where $\overline{\fk}$ is the algebraic closure of $\fk$.

\begin{proposition}\label{prop: triangular}
  Let $T\subset \kx$ be a polynomial set in the form
  $$[t_1(x_1), x_2 - t_2(x_1), \ldots, x_n-t_n(x_1)], $$
  and $t_1=t_{1,1} \cdots t_{1,r}$. For $i=1, \ldots, r$, define
\begin{equation*}
T(i) = [t_{1,i}, x_2-t_{2,i}, \ldots, x_n - t_{n,i}],
\end{equation*}
where $t_{j, i}$ is the remainder of $t_j$ modulo $t_{1,i}$ for $j=2, \ldots, n$. Then we have the following zero relation
\begin{equation}\label{et:ideal-relation}
\zero(T) = \bigcup_{i=1}^r \zero(T(i)).
\end{equation}
\end{proposition}

Let $\overline{f}_1$ be the squarefree part of $f_1$. As each $P_i$ in \eqref{eq: pi} satisfies the form in Proposition \ref{prop: triangular}, we can compute $t$ new polynomial sets $\overline{P}_j$ whose univariate polynomials in $x_1$ is $\overline{f}_{1,j}$ for $j=1, \ldots, t$, such that $\overline{f}_1 = \prod_{j=1}^t \overline{f}_{1,j}$, and $\overline{f}_{1, j}$ are pairwise coprime. These new polynomial sets can be found in the following way. Set $p = \overline{f}_1$. We start with $j=1$ and computes $\overline{f}_{1,j} =\gcd(f_{1,j}, p)$. As long as this polynomial is not equal to $1$, a new polynomial set $\overline{P}_j$ whose univariate polynomial is $\overline{f}_{1,j}$ is constructed from $P_j$ by Proposition \ref{prop: triangular}. Next set $p := p/ \overline{f}_{1,j}$ and check whether $p= 1$. If so, we know we already have enough new polynomial sets; otherwise $j := j+1$, and the process above is repeated.

Now we reduce the current case to the earlier one with $\overline{f}_1$ squarefree and $\overline{P}_1, \ldots, \overline{P}_t$ to construct the modulo equation sets. Thus the \grobner basis of $\sqrt{I}$ w.r.t.\ LEX can be obtained similarly (note that extracting the squarefree part of $f_1$ results in the radical of $I$).

The whole method based on the deterministic Wiedemann algorithm is summarized in Algorithm \ref{alg:deterministic} below. The subfunction $\sqrfree()$ returns the squarefree part of the input polynomial. The operator ``$\cat$" means concatenating two sequences.

\begin{algorithm}[!h]\label{alg:deterministic}%\linesnumbered
     \KwIn{$G_1$, \grobner basis of a 0-dimensional ideal $I \subset \kx$ w.r.t.\ $<_1$}
     \KwOut{$G_2$, \grobner basis of $\sqrt{I}$ w.r.t.\ LEX if $I$ is in shape position;
     {\sf Fail}, otherwise.}
     \BlankLine

     Compute the canonical basis of $\kx/\ideal{G_1}$ and multiplication matrices $T_1, \ldots, T_n$\;
     $\p{e}_1 = (1, 0,\ldots, 0)^t, \p{e}_2 = (0,1, 0, \ldots, 0)^t, \ldots, \p{e}_D = (0, \ldots, 0, 1)^t \in \fk^{(D \times 1)}$\;
     $k := 1$; ~$F := [\,]$; ~$f ;= 1$; ~$d := 0$; ~$\p{b} = \p{e}_1$; ~$S =[\,]$\;
     \While{$\p{b} \neq \p{0}$}
        {
            $s := [\inner{\p{e}_k, T_1^i \p{b}}: i = 0, 1, \ldots, 2(n-d)-1]$\;
            $g := \BM(s)$\;
            $f := f \cdot g$; ~$d := \deg(f)$; ~$F := F \cat~[g]$; ~$\p{b} := g(T_1)\p{b}$; ~$S := S \cat~[s]$\;
            $k := k+1$\;
        }

     (Suppose $F = [f_{1,1}, \ldots, f_{1, r}]$) $f_1 := \prod_{i=1}^r f_{1,i}$\;

     \eIf{$\deg(f_1) \neq D$}
     {\Return{{\sf Fail}}}
     {
     \For{$i=1, \ldots, r$}
     {
        $d_i := \deg(f_{1, i})$\;
        \For{$j=2, \ldots, n$}
        {
            Construct the Hankel matrix $H_j$ and the vector $\p{b}$ from $S$\;% to construct the linear equation $H_j \p{x} = \p{b}$\;
            Compute $\p{c} = (c_1, \ldots, c_{d_i})^t := H_j^{-1} \p{b}$;  ~$f_{i, j} := \sum_{k=0}^{d_i}\p{c}_{k+1}x_1^k$\;
        }
     }

     $\overline{f}_1 := \sqrfree(f_1)$\;
     \If{$\overline{f}_1 \neq f_1$}
     {
        Compute $\{[\overline{f}_{1,j}, x_2 - \overline{f}_{2, j}, \ldots, x_n-\overline{f}_{n, j}]: j = 1, \ldots, t\}$ from $\{[f_{1,i}, x_2 - f_{2, i}, \ldots, x_n-f_{n, i}]: i = 1, \ldots, r\}$ by Proposition \ref{prop: triangular} such that $\overline{f}_1 = \prod_{j=1}^t \overline{f}_{1,j}$ and $\overline{f}_{1, j}$ are pairwise coprime\;
     }\label{line: splitting}

     \For{$j=2, \ldots, n$}
     {
        Solve the modulo equation set \eqref{eq: crt} to get $f_j$\;
     }

     \Return{$[\overline{f}_1, x_2-f_2, \ldots, x_n-f_n]$}
     }

\caption{Shape position (deterministic) $G_2 := \shaped(G_1, <_1)$}
\end{algorithm}

\begin{remarks}\label{rm: deterministic}
  If the factors $f_{1, 1}, \ldots, f_{1,r}$ of $f_1$ returned by the deterministic Wiedemann algorithm are pairwise coprime (which needs extra computation to test), the \grobner basis of $I$ w.r.t.\ LEX can be computed from the {\sf CRT}.

  The method of the deterministic version described above is also applicable to the Wiedemann algorithm with several random vectors. To be precise, when the first random vector does not return the correct polynomial $f_1$, one may perform a similar procedure as the deterministic Wiedemann algorithm by updating the sequence with a newly chosen random vector (instead of $\p{e}_i$ in the basis) and repeating \cite{W86S}. In that case, the method above with {\sf CRT} can also be used to compute the \grobner basis of $\sqrt{I}$ w.r.t.\ LEX.
\end{remarks}

\subsubsection{Complexity}
Next the computational complexity, namely the number of field operations needed, for the deterministic method for ideals in shape position is analyzed.

(1) In total the deterministic Wiedemann algorithm needs
$$O(D(N_1 + D\log(D)\log\log(D)))$$
 operations if fast polynomial multiplications are used \cite{W86S}. Here $N_1$ still denotes the number of nonzero entries in $T_1$.

(2) Next at most $D$ structured linear equation sets with Hankel coefficient matrices are constructed and solved, each with maximum operations $O(D\log(D))$. Hence this procedure needs $O(D^2 \log(D))$ operations at most. % (note that for each linear system, the size of the Hankel matrix is smaller than $D\times D$)

(3) The squarefree part $\overline{f}_1$ of $f_1$ can be obtained with complexity at most $O(D^2\log(D))$ for the case when $\field{K}$ has characteristic $0$ and $O(D^2 \log(D) + D \log(q/p))$ for characteristic $p>0$ respectively, where $|\field{K}| = q$ \cite[Theorem 14.20 and Exercise 14.30]{VG2003M}. For the case when $f_1$ is not squarefree, suppose $t$ new polynomial sets $\overline{P_{1}}, \ldots, \overline{P_{t}}$ are needed, and $\deg(\overline{f}_{1,i}) =d_i$ for $i=1, \ldots, t$. To compute each set $\overline{P}_i$ of the form \eqref{eq: shape}, $n-1$ polynomial divisions are needed to find the remainders, with complexity $O(nd_iD)$. Hence the total complexity to obtain $\overline{P}_1, \ldots, \overline{P}_t$ is
$$O(\sum_{i=1}^t nd_iD) = O(nD \sum_{i=1}^t d_i) \leq O(nD^2),$$
for we have $\sum_{i=1}^t d_i =\deg(\overline{f}_1) < D$.

(4) Solving the modulo equation set \eqref{eq: crt} for each $j=2, \ldots, n$ requires $O(D^2)$ operations at most by \cite[Theorem 5.7]{VG2003M}. Thus in total $O(nD^2)$ operations are needed for the {\sf CRT} application.

Therefore, we have the following complexity result for the method with the deterministic Wiedemann algorithm.
\begin{theorem}
  Assume that $T_1$ is known. If the input ideal $I$ is in shape position, then this deterministic method will return the \grobner basis of $\sqrt{I}$ w.r.t.\ LEX with the complexity
  $$O(D(N_1 + D (\log(D) \log\log(D) + n))).$$
\end{theorem}

%For each $i = 1, \ldots, D$, one can construct the linear recurring sequence
%$$[\inner{\p{e}_i, T_1^j \p{e}}: j = 0, \ldots, 2D-1]$$
%and compute its minimal polynomial $f_{1, i}$ with the Berlecamp--Massey algorithm. Then the polynomial $f_1$ can be obtained by computing the least common multiplier of $f_{1,1}, \ldots, f_{1, D}$, i.e. $f_1 = \lcm(f_{1,1}, \ldots, f_{1, D})$.
%
%Yet a smarter or faster strategy is to update a vector $\p{b}$ after each $f_{1, i}$ is computed, instead of using the same one $\p{e}$ throughout. After the minimal polynomial $f_{1,1}$ is computed with the sequence $[\inner{\p{e}_1, T_1^j \p{e}}: j = 0, \ldots, 2D-1]$, a new vector $\p{b}= f_{1,1}(T_1)\p{e}$ is generated. If $\p{b}=\p{0}$,
% as , and also a new linearly recurring sequence
%$$[\inner{\p{e}_2, T_1^j \p{b}}: j = 0, \ldots, 2D-1].$$
%It is not hard to verify that the minimal polynomial $f_{1,2}$ of this sequence is a factor of $f_1/f_{1,1}$. Repeating the above procedure for $i=3, \ldots, D$, one also gets $f_{1, 3}, \ldots, f_{1, D}$ and $f_1$ is just the product of $f_{1,1}, f_{1,2}, \ldots, f_{1,D}$.

\subsubsection{Example}
Here is a  toy example to illustrate how the deterministic method works. Consider an ideal $I$ in $\ff_2[x_1, x_2]$ generated by its \grobner basis w.r.t.\ DRL
$$G_1 :=[ x_2x_1^3+x_1^3 +x_1+1, x_1^4+x_1^3+x_2+1, x_1^2+x_2^2].$$
Its \grobner basis w.r.t.\ LEX is
$$G_2 = [ f_1 := (x_1+1)^3 (x_1^2+x_1+1)^2, x_2 + x_1^4 + x_1^3+1],$$
from which one can see that $I$ is in shape position.

From $G_1$ the canonical basis $B = [1, x_1, x_2, x_1^2, x_1x_2, x_1^3, x_1^2x_2]$ and the multiplication matrices $T_1$ and $T_2$ are first computed. With a vector $\p{r} = (1, 1, 0, 1, 0, 1, 0)^t \in \ff_2^{(7\times 1)}$ generated at random, the classical Wiedemann algorithm will only return a proper factor $(x_1+1)(x_1^2+x_1+1)$ of $f_1$, and whether $I$ is in shape position is unknown.

Next we use the deterministic Wiedemann algorithm to recover $f_1$. With $\p{e}_1 = (1, 0,\ldots, 0)^t$, a factor $f_{1,1} = (x_1+1)^2(x_1^2+x_1+1)$ of $f_1$ is found with the Berlekamp--Massey applied to the sequence \eqref{eq: deterministic-1}. Then we update the vector
$$\p{b} = f_{1,1}(T_1)\p{e} = (0,1,1,0,0,0,0)^t,$$
and execute the second round with $ \p{e}_2 = (0,1, 0, \ldots, 0)^t$, obtaining another factor $f_{1,2} = (x_1+1)(x_1^2+x_1+1)$. This time the updated vector $\p{b} = \p{0}$, thus the deterministic Wiedemann algorithm ends, and $f_1$ is computed as $f_{1,1}f_{1,2}$. As $\deg(f_{1,1}f_{1,2}) = D$, now $I$ is verified to be in shape position.

Then we construct the linear equation sets similar to \eqref{eqn:Hankel} to recover $f_{2,1}$ and $f_{2,2}$ respectively. The first one, for example, is
$$
  \left(\begin{tabular}{cccc}
1 & 0 & 0 &0 \\
0 &0 &0&1\\
0&0&1&1 \\
0&1&1&1
\end{tabular}
\right)
\cdot
\left(\begin{tabular}{c}
$c_0$\\$c_1$\\$c_2$\\$c_3$
\end{tabular}
\right)
=
\left(\begin{tabular}{c}
0\\0\\0\\1
\end{tabular}
\right).
$$
After solving them,  we have the \grobner bases of $I+\ideal{f_{1,1}}$ and $I+\ideal{f_{1,2}}$ respectively as
\begin{equation*}
  \begin{split}
    P_1 &= [(x_1+1)^2(x_1^2+x_1+1), x_2+x_1], \\
    P_2 &= [(x_1+1)(x_1^2+x_1+1), x_2+x_1].
  \end{split}
\end{equation*}

Then the squarefree part $\overline{f}_1$ of $f_1$ is computed, and we find that $I$ is not radical, and thus only the \grobner basis $\tilde{G}_2$ of $\sqrt{I}$ w.r.t.\ LEX may be computed. From $f_{1,2} = \overline{f}_1$, we directly have $\tilde{G}_2 = P_2$, and the algorithm ends.

The way to compute $\tilde{G}_2$ by {\sf CRT}, which is more general, is also shown in the following. Two new polynomial sets
\begin{equation*}
    \overline{P}_1 = [x_1+1, x_2+1], \quad \overline{P}_2 = [x_1^2+x_1+1, x_2+x_1]
\end{equation*}
are first computed and selected according to $\overline{f}_1$ by Proposition \ref{prop: triangular}. Then the modulo equation set
\begin{equation*}
\left\{
  \begin{split}
    f_2 &\equiv x_2+1 ~\,\mod x_1+1, \\
    f_2 &\equiv x_2+x_1 \mod x_1^2+x_1+1
  \end{split}
  \right.
\end{equation*}
as \eqref{eq: crt} is solved with {\sf CRT}, resulting in the same $\tilde{G}_2$. One can check that $\tilde{G}_2$ is the \grobner basis of $\sqrt{I}$ w.r.t.\ LEX with any computer algebra system.

\subsection{Incremental algorithm to compute the univariate polynomial}
For a 0-dimensional ideal $I\subset \kx$, the univariate polynomial in its \grobner basis w.r.t.\ LEX is of special importance. For instance, it may be the only polynomial needed to solve some practical problems. Furthermore, in the case when $\fk$ is a finite field, after the univariate polynomial is obtained, it will not be hard to compute all its roots, for one can simplify the original polynomial system by substituting the roots back, and sometimes the new system will become quite easy to solve.
%(for instance ideals in shape position or when the degree of the univariate polynomial is close to that of the ideal).

Besides the two methods in the previous parts, next the well-known incremental Wiedemann algorithm dedicated to computation of the univariate polynomial is briefly recalled and discussed.

In the Wiedemann algorithm, the dominant part of its complexity comes from construction of the linearly recurring sequence ($O(DN_1)$), while the complexity of the Berlekamp--Massey algorithm is relatively low ($O(D\log(D))$). Hence the idea of the incremental method is to construct the sequence incrementally to save computation and apply the Berlekamp--Massey algorithm to each incremental step.

We start with the linearly recurring sequence $[\inner{\p{r}, T_1^i \p{e}}: i = 0, 1]$ and compute its minimal polynomial with the Berlekamp--Massey algorithm. Next we proceed step by step with the sequence
$$[\inner{\p{r}, T_1^i \p{e}}: i = 0, \ldots, 2k-1]$$
until the returned polynomial coincides with the one in the previous step. Then this minimal polynomial equals the univariate polynomial $f$ we want to compute with a large probability.

Suppose $\deg(f)=d$. Then the number of steps the method takes is bounded by $d+1$. In other words, the method stops at most after the sequence $[\inner{\p{r}, T_1^i \p{e}}: i = 0, \ldots, 2d+1]$ is handled. The number of field operations to construct the sequences is $O(dN_1)$, while the total complexity to compute the minimal polynomials with the Berlekamp--Massey algorithm is $O(\sum_{k=1}^{d+1}k^2 )= O(d^3)$ (note that in the incremental case, the fast Berlekamp--Massey with complexity $O(k\log(k))$ is not applicable). Therefore the overall complexity for the incremental Wiedemann method to compute the univariate polynomial is $O(dN_1+d^3)$. As can be seen here from this complexity, this incremental method is sensitive to the output polynomial $f$. When the degree $d$ is relatively small compared with $D$, this method will be useful.

\section{General case: {\sf BMS}-based algorithm}\label{sec:general}

In the general case when the ideal $I$ may not be in shape position, perhaps those methods described in Section \ref{sec:shape} will not be applicable. However, we still want to follow the idea of constructing linearly recurring sequences and computing their minimal polynomials with the Berlekamp--Massey algorithm. The way to do so is to generalize the linearly recurring sequence to a multi-dimensional linearly recurring relation and apply the {\sf BMS} algorithm to find its minimal generating set.

\subsection{Algorithm description}\label{subsec: general-1}
We first define a $n$-dimensional mapping $E: \znum_{\geq 0}^n \longrightarrow \fk$ as
\begin{equation}\label{eq: mapping}
  (s_1, \ldots, s_n) \longmapsto \ideal{\p{r}, T_1^{s_1}\cdots T_n^{s_n} \p{e}},
\end{equation}
where $\p{r}\in \fk^{(D\times 1)}$ is a random vector. One can easily see that such a mapping is a $n$-dimensional generalization of the linearly recurring sequence constructed in the Wiedemann algorithm.

Note that $T_1^{s_1}\cdots T_n^{s_n} \p{e}$ in the definition of $E$ above is the coordinate vector of $(s_1, \ldots, s_n)$ in the {\sf FGLM} algorithm. As a polynomial $f$ in the \grobner basis of $I$ is of form \eqref{eq: fglm-gb}, and the linear dependency \eqref{eq: fglm-linear} holds, one can verify that $f$ satisfies \eqref{eqn:charPoly} and thus is a polynomial in $I(E)$. The {\sf BMS} algorithm is precisely the one to compute the \grobner basis of $I(E)$ w.r.t.\ to a term ordering, so one may first construct the mapping $E$ via $T_1, \ldots, T_n$, and attempts to compute the \grobner basis of $I$ from the {\sf BMS} algorithm applied to $I(E)$.

We remark that $f$ is in $I(E)$ for any vector $\p{r}$. In fact, the idea above is a multi-dimensional generalization of the Wiedemann algorithm. The minimal polynomial $g$ of the Krylov sequence $[\p{b}, A\p{b}, A^2 \p{b}, \ldots]$ is what the Wiedemann algorithm seeks, for $g$ directly leads to a solution of the linear equation $A\p{x} = \p{b}$ for a nonsingular matrix $A$ and vector $\p{b}$. Then a random vector is chosen to convert the sequence to a scalar one
$$[\inner{\p{r}, \p{b}}, \inner{\p{r}, A\p{b}}, \inner{\p{r}, A^2 \p{b}}, \ldots],$$
and the Berlekmap--Massey algorithm is applied to find the minimal polynomial of this new sequence, in the hope that $g$ can be obtained. While the method proposed here converts the mapping from $(s_1, \ldots, s_n)$ to its coordinate vector in the {\sf FGLM} to a $n$-dimensional scalar mapping with a random vector, and then the {\sf BMS} algorithm (generalization of Berlekamp--Massey) is applied to find the minimal polynomial set, which is also the \grobner basis, w.r.t.\ to a term ordering.

This method for computing the \grobner basis of $I$ makes full use of the sparsity of $T_1, \ldots, T_n$, in the same way as how the Wiedemann algorithm takes advantage of the sparsity of $A$. The method is a probabilistic one, also the same as the Wiedemann algorithm. This is reasonable for the ideal $I(E)$ derived from the $n$-dimensional mapping may lose information of $I$ because of the random vector, with $I \subset I(E)$. Clearly, when $I$ is maximal (corresponding to the case when $g$ in the Wiedemann algorithm is irreducible), $I(E)$ will be equal to $I$. Furthermore, as polynomials in the \grobner basis are characterized by the linear dependency in \eqref{eq: fglm-linear}, we are always able to check whether the \grobner basis of $I(E)$ returned by the {\sf BMS} algorithm is that of $I$.
%, and in fact the ideal generated by a random polynomial set is maximal in most cases

\begin{remarks}
  When the term ordering in the {\sf BMS} algorithm is LEX, computation of the univariate polynomial in this method is exactly the same as that described in Section \ref{subsec: shape-normal}. This is true because for the LEX ordering ($x_1 < \cdots < x_n$), the terms are ordered as $$[1, x_1, x_1^2, \ldots, x_2, x_1x_2, x_1^2x_2, \ldots ],$$
hence the first part of $E$ is $E((p_1, 0, \ldots, 0))=\ideal{\p{r}, T_1^{p_1} \p{e}}$,
and the {\sf BMS} algorithm degenerates to the Berlekamp--Massey one.
\end{remarks}

Another fact we would like to mention is that the {\sf BMS} algorithm from Coding Theory is mainly designed for graded term orderings like DRL, for such orderings are \emph{Archimedean} and have good properties to use in algebraic decoding \cite{CLO1998U}. But it also works for other orderings, though extra techniques not contained in the original literature have to be introduced for orderings dependent on LEX (like LEX itself or block orderings which break ties with LEX).

Take the term ordering LEX for instance, an extra polynomial reduction is performed after every $\sakata()$ step to control the size of intermediate polynomials. This is actually not a problem for orderings like DRL, for in that case the leading term of a polynomial will give a bound on the size of terms in that polynomial. We also have to add an extra termination check for each variable $x_i$, otherwise the {\sf BMS} algorithm will endlessly follow a certain part of the terms. For example, all variables in the sequence $[1, x_1, x_1^2, \ldots]$ are smaller than $x_2$, and the original {\sf BMS} does not stop handling that infinite sequence by itself.

With all the discussions, the algorithm is formulated as follows. The ``Termination Criteria" here in this description mean that $F$ does not change for a certain number of iterations. The subfunction $\reduce(F)$ performs reduction on $F$ so that every polynomial $f \in F$ is reduced w.r.t.\ $F\setminus \{f\}$, and  $\IsGB(F)$ returns \textsf{true} if $F$ is the \grobner basis of $I$ w.r.t.\ LEX and \textsf{false} otherwise.

\begin{algorithm}[h]\label{alg: bms}%\linesnumbered
    \KwIn{$G_1$, \grobner basis of a 0-dimensional ideal
    $I \subset \kx$ w.r.t.\ $<_1$}
    \KwOut{\grobner basis of $I$ w.r.t.\ $<_2$; or {\sf Fail}, if the {\sf BMS} algorithm fails returning the correct \grobner basis}

\BlankLine

Compute the canonical basis of $\kx/\ideal{G_1}$ and multiplication matrices $T_1, \ldots, T_n$\;

Choose $\p{r}\in \fk^{(D\times 1)}$ at random\;

$\p{u}:=\p{0}$; ~$F := [1]$; ~$G:= [~]$; ~$E:=[~]$\;

 \Repeat{Termination Criteria}
{\label{line:sakata-begin}
    $e := \ideal{\p{r}, T_1^{\p{u}_1}\cdots T_n^{\p{u}_n} \p{e}}$\label{line: mapping}\;

    $E := E ~\cat~ [e]$\;
    $F, G := \sakata(F, G, \p{u}, E)$\;
    $\p{u} := \Next(\p{u})$ w.r.t.\ $<_2$\;
    $F := \reduce(F)$\;
}\label{line:sakata-end}

\eIf{$\IsGB(F)$}{\Return{F}}{\Return{{\sf Fail}}}
\caption{General case $G_2 := \BMSbased(G_1, <_1)$}\label{alg: bmsbased}
\end{algorithm}

%\textbf{Correctness}. After the termination criteria are reached, the main loop ends and whether the returned polynomial set $F$ is the \grobner basis of $I$ w.r.t.\ LEX is tested. If $\IsGB(F)=\sf{true}$, $F$ is already the \grobner basis we want to compute and the algorithm naturally finishes. While $\IsGB(F)=\sf{false}$ means that the {\sf BMS} algorithm returns a polynomial set $F$ which is only \grobner basis of $I(E)$, but $I(E)\neq I$ (on the assumption that the termination criteria work). Then we have to return to the original {\sf FGLM} algorithm to complete the change of ordering.

The correctness of Algorithm \ref{alg: bms} is obvious. Next we prove its termination. Once the loop ends, the algorithm almost finishes. Hence we shall prove the termination of this loop. Clearly when the polynomial set $F$ the {\sf BMS} algorithm maintains turns to the \grobner basis of $I(E)$ w.r.t.\ $<_2$, the current termination criterion, namely $F$ keeps unchanged for a certain number of passes, will be satisfied. And a sufficient condition for $F$ being the \grobner basis is given as Theorem \ref{thm: termination} below.

\subsection{Complexity}

Part of earlier computation of values of $E$ can be recorded to simplify the computation at Line \ref{line: mapping}. Suppose the value of $E$ at a certain term $(u_1, u_2, \ldots, u_{i-1}, u_i-1, u_{i+1}, \ldots, u_n)$
$$\tilde{\p{e}} = T_1^{\p{u}_1}\cdots T_i^{\p{u}_i-1} \cdots T_n^{\p{u}_n} \p{e}$$
has been computed and recorded. Then we know the value at $\p{u} = (u_1, \ldots, u_n)$ is
$$\inner{\p{r}, T_1^{\p{u}_1}\cdots T_n^{\p{u}_n} \p{e}} = \inner{\p{r}, T_i \tilde{\p{e}}}, $$
for all $T_i$ and $T_j$ commute. Thus the computation of one value of $E$ can be achieved within $O(N)$ operations, where $N$ is the maximal number of nonzero entries in matrices $T_1, \ldots, T_n$.

Next we focus on the case when the target term ordering is LEX. The complexities of the three steps in Algorithm \ref{alg:sakata} are analyzed below.

(1) As an extra reduction step is applied after each iteration, the numbers of terms of polynomials in $F$ are bounded by $D+1$. Denote by $\hat{N}$ the number of polynomials in $G_2$. Then checking whether $F$ is valid up to $\Next(\p{u})$ needs $O(\hat{N}D)$ operations.

(2) The computation of the new delta set $\Delta(\Next(\p{u}))$ only involves integer computations, and thus no field operation is needed.

(3) Constructing the new polynomial set $F^+$ valid up to $\Next(\p{u})$ requires $O(\hat{N}D)$ operations at most. The readers may refer to \cite{SH2002A, CLO1998U} for the way to construct new polynomials.

In step (1) above, new values of $E$ other than $e$ may be needed for the verification. The complexity for computing them is still $O(N)$, and this is another difference from the original {\sf BMS} algorithm for graded term orderings. After the update is complete, a polynomial reduction is applied to $F$ to control the size of every polynomial. This requires $O(\hat{N}\bar{N}D)$ operations, where $\bar{N}$ denotes the maximum term number of polynomials in $G_2$. To summarize, the total operations needed in each pass of the main loop in Algorithm \ref{alg: bms} is
$$O(N+\hat{N}D+\hat{N}\bar{N}D)=O(N+\hat{N}\bar{N}D).$$
Hence to estimate the whole complexity of the method, we only need an upper bound for the number of passes it takes in the main loop.

\begin{theorem}\label{thm:loop}\label{THM:LOOP}
  Suppose that the input ideal $I \subset \kx$ is of degree $D$. Then the number of passes of the loop in Algorithm \ref{alg: bms} is bounded by $2nD$.
\end{theorem}

Before giving the proof, we need to introduce some of the proven results on the {\sf BMS} algorithm for preparations. Refer to \cite{BO2006C, CLO1998U} for more details.

Denote the previous term of $\p{u}$ w.r.t.\ $<$ by $\pre(\p{u})$. Given a $n$-dimensional array $E$, suppose now a polynomial $f \in \kx$ is valid for $E$ up to $\pre(\p{u})$ but not to $\p{u}$. Then the term $\p{u} - \lt(f)$ is called the \emph{span} of $f$ and denoted by $\Span(f)$, while the term $\p{u}$ is called the \emph{fail} of $f$ and written as $\fail(f)$. When $f \in I(E)$, $f$ is valid up to every term, and in this case we define $\Span(f) := \infty$. The following proposition reveals the importance of spans.

\begin{proposition}[\,{\cite[Corollary 9]{BO2006C}}]
 $\Delta(E) = \{\Span(f): f \not \in I(E)\}$.
\end{proposition}

Define $I(\p{u}) := \{ f \in \kx: \fail(f) > \p{u}\}$. Such a set is not an ideal but is closed under monomial multiplication: supposing that $F\in \ideal{a}(\p{u})$, we have $\p{t}F \in \ideal{a}(\p{u})$ for every term $\p{t}\in \kx$.

\begin{proposition}[\,{\cite[Proposition 6]{BO2006C}}]\label{prop: deltaset}
  For each $\p{u}$, $\Delta(\p{u}) = \{\Span(f): f \not \in I(\p{u})\}$. Furthermore, $\p{v} \in \Delta(\p{u})\setminus \Delta(\pre(\p{u}))$ if and only if $\p{v} \prec \p{u}$ and $\p{u} - \p{v} \in \Delta(\p{u})\setminus \Delta(\pre(\p{u}))$.
\end{proposition}

The above proposition states when a term in $\Delta(E)$ is determined, and it is going to be used extensively in the sequel. Also from this proposition, one can derive the following termination criteria for the {\sf BMS} algorithm, which are mainly designed for graded term orderings like DRL.

\begin{theorem}[\,{\cite[pp.529, Proposition (3.12)]{CLO1998U}} ]\label{thm: termination}
  Let $c_{max}$ be the largest element of $\Delta(E)$ and $s_{max}$ be the largest element of $\{\lt(g): g \in G\}$, where $G$ is the \grobner basis of $I(E)$ w.r.t.\ $<$.
  \begin{enumerate}
    \item[(1)] For all $\p{u} \geq c_{max} + c_{max}$, $\Delta(\p{u}) = \Delta(E)$ holds.
    \item[(2)] For all $\p{v} \geq c_{max} + \max\{c_{max}, s_{max}\}$, the polynomial set $F$ the {\sf BMS} algorithm maintains equals $G$.
  \end{enumerate}
\end{theorem}

As explained in Section \ref{subsec: general-1}, actually the term ordering LEX is not the one of interest in Coding Theory and does not possess some properties needed for a good order domain.  But the results stated above are still correct. In particular, Theorem \ref{thm: termination} indicates when the iteration in the {\sf BMS} algorithm ends. For graded term orderings like DRL, once the termination term is fixed, the whole intermediate procedure in the {\sf BMS} algorithm is also determined. However, for LEX it is not the case. We have to study carefully what happens between the starting term $\p{1}$ and the termination term indicated by Theorem \ref{thm: termination}.

Next we first illustrate the procedure for a 2-dimensional example derived from Cyclic5. Both the delta set (marked with crosses) and the terms handled by the {\sf BMS} algorithm (with diamonds) are shown in Figure \ref{fig:process}.

\begin{figure}
\centering
\includegraphics[width=10cm]{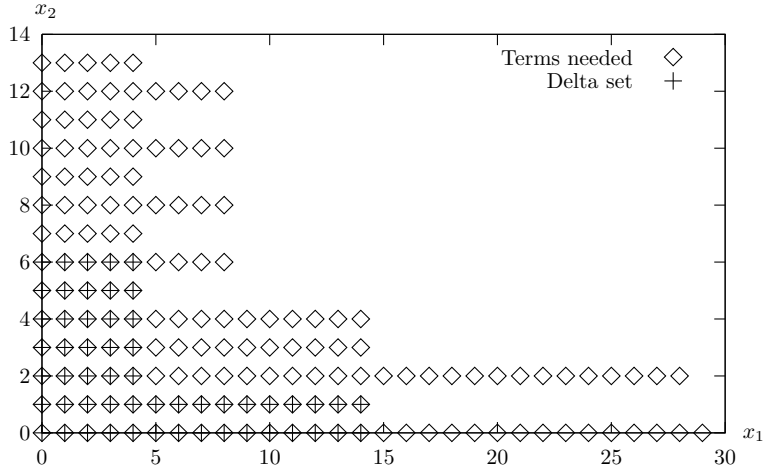}
\caption{Delta set ($+$) and terms needed ($\diamond$) for Cyclic5-2}\label{fig:process}
\end{figure}

The $c_{max}$ and $s_{max}$ in Theorem \ref{thm: termination} are respectively $(4,6)$ and $(0,7)$. In fact, the {\sf BMS} algorithm obtains the whole delta set at $(8,12) = c_{max}+c_{max}$, and the polynomial set it maintains grows to the \grobner basis at $(4, 13) = c_{max}+s_{max}$, which is also where the algorithm ends.

Next we go into some details of what happens when a diamond row is handled by the {\sf BMS} algorithm. We call a diamond (or cross) row the $j$th diamond (or cross) row if terms in this row are $(i, j)$. Then for the $0$th diamond row, the {\sf BMS} algorithm degenerates to the Berlekamp--Massey one to compute the univariate polynomial $f_1(x_1)$. Here $30$ diamond terms are needed because the minimal polynomial is of degree $15$.

For other rows in Figure \ref{fig:process}, from Proposition \ref{prop: deltaset}, one knows that at a $j$th diamond rows with an odd $j$, the delta set does not change. Thus such diamond rows are only bounded by the latest verified row in the delta set. This is because otherwise a wrong term in the delta set will be added if other diamond terms are handled. For example, the $3$rd diamond row is of the same length as the $1$st cross row, while the $5$th diamond row is as that of the $2$nd cross one.

For  a $2k$th diamond row, its number is related to two cirteria. On one hand, again from Proposition \ref{prop: deltaset}, the $k$th cross row is determined while the $2k$th diamond row is handled in the {\sf BMS} algorithm. Denote by  $c_{max}(k)$ the largest term in the $k$th cross row, then terms up to $c_{max}(k)+c_{max}(k)$ in the $2k$th diamond row have to be handled to furnish the $k$th cross row. On the other hand, the number of $2k$th diamond row is also bounded by the latest verified cross row, as the odd diamond ones. The first criterion is shown by the $6$th diamond and the $3$rd cross rows, while the $4$th diamond row is the result of both criteria.

For a term $\p{u} = (u_1, \ldots, u_i, 0, \ldots, 0) \in \fk[x_1, \ldots, x_i] $, in the proof below we write it as $\p{u} = (u_1, \ldots, u_i)$ for simplicity, ignoring the last $n-i$ zero components in the terms.

\begin{proof}(of Theorem \ref{THM:LOOP})
Suppose $G$ is the \grobner basis of $I(E)$ the {\sf BMS} algorithm computes. Denote the number of terms needed in the {\sf BMS} algorithm to compute $G \cap \fk[x_1, \ldots, x_i]$ by $\chi_i$, and $\Delta_i := \Delta(E) \cap \fk[x_1, \ldots, x_i]$. From $I\subseteq I(E)$ one knows that $\Delta(E)$ is a subset of the canonical basis of $\kx/I$, thus $|\Delta(E)| \leq D$. Therefore to prove the theorem, it suffices to show $2n|\Delta(E)|$ is an upper bound.

We induce on the number of variable $i$ of $\fk[x_1, \ldots, x_i]$. For $i=1$, the {\sf BMS} algorithm degenerates to the Berlekamp--Massey, and one can easily see that $\chi_1 \leq 2|\Delta_1|$ holds. Now suppose $\chi_k\leq 2k|\Delta_k|$ for $k (<n)$. Next we prove $\chi_{k+1}\leq 2(k+1)|\Delta_{k+1}|$.

As previously explained, in the terms to compute $G\cap \fk[x_1, \ldots, x_{k+1}]$, the terms $(u_1, \ldots, u_k, 2l)$ are determined by two factors: terms $(v_1, \ldots, v_k, l)$ in $\Delta_{k+1}$, and the latest verified terms in the delta set. First we ignore those $(u_1, \ldots, u_k, 2l)$ terms determined by the latter criterion, and denote by $\mathcal{T}_{k+1}$ all the remaining ones in $\fk[x_1, \ldots, x_{k+1}]$. We claim that $|\mathcal{T}_{k+1}|$ is bounded by $(2k+1)|\Delta_{k+1}|$.

From Theorem \ref{thm: termination}, we can suppose there exists some integer $m$, such that
  \begin{equation}\label{eq: proof}
      \mathcal{T}_{k+1} = \bigcup_{l=0, \ldots, 2m+1} \mathcal{T}_{k+1, l}, \quad \Delta_{k+1} = \bigcup_{j=0, \ldots, m} \Delta_{k+1, j},
  \end{equation}
where
  \begin{equation*}
    \begin{split}
      \mathcal{T}_{k+1, l} &:= \{\p{u} \in \mathcal{T}_{k+1} : \p{u}= (u_1, \ldots, u_k, l)\},\\
      \Delta_{k+1, j} &:= \{\p{u}\in \Delta_{k+1}: \p{u}= (u_1, \ldots, u_k, j)\} .
    \end{split}
  \end{equation*}

Clearly $|\mathcal{T}_{k+1,0}| = \chi_k \leq 2k|\Delta_k|$, and $\Delta_{k+1, 0} = \Delta_k$. One can see that $|\mathcal{T}_{k+1, 2j}|$ is bounded by either $2k|\Delta_k| = 2k|\Delta_{k+1, 0}|$ (if $j=0$) or $2|\Delta_{k+1, j}|~(\leq 2k|\Delta_{k+1, j}|)$.
Furthermore, $|\mathcal{T}_{k+1, 2j+1}|$ is bounded by $|\Delta_{k+1, j}|$, the number of the latest verified delta set. Hence we have
  $$|\mathcal{T}_{k+1, 2j}|+|\mathcal{T}_{k+1, 2j+1}|\leq (2k+1)|\Delta_{k+1, j}|,$$
which leads to $|\mathcal{T}_{k+1}|\leq (2k+1)|\Delta_{k+1}|$.

Now we only need to show the number of all the previously ignored terms, denoted by $\mathcal{T}'_{k+1}$, is bounded by $|\Delta_{k+1}|$. Suppose $\mathcal{T}'_{k+1} = \bigcup_{i \in S} \mathcal{T}'_{k+1, i}$, where $S$ is a set of indexes with $|S| \leq m$ in \eqref{eq: proof}, and
$$\mathcal{T}'_{k+1, i} := \{\p{u}\in \mathcal{T}'_{k+1}: \p{u} = (u_1, \ldots, u_k, i)\}.$$
Then for each $i$, $|\mathcal{T}'_{k+1, i}|$ is bounded by the number of the latest verified delta set, say $|\Delta_{k+1, p_i}|$. Thus the conclusion can be proved if one notices $\bigcup_{i\in S}\Delta_{k+1, p_i} \subseteq \Delta_{k+1}$.
\end{proof}

\begin{theorem}\label{prop:complexity general}
Assume that $T_1, \ldots, T_n$ are constructed. The complexity for Algorithm \ref{alg: bms} to complete the change of ordering is bounded by $O(nD(N+\hat{N}\bar{N}D))$, where $N$ is the maximal number of nonzero entries in the multiplication matrices $T_1, \ldots, T_n$, and $\hat{N}$ and $\bar{N}$ are respectively the number of polynomials and the maximal term number of all polynomials in the resulting \grobner basis.
\end{theorem}
%This can be easily proved with the observation that for each $j$, the
%
%Suppose these cases occur at $(\p{p}_l, l), l=i_1, \ldots, i_{m'}$. Again from the procedure, one can see that the number of terms in $(Q, i)$ for $(\p{p}_1, i_1)$ is bounded by $|\Delta_k|$. And after it occurs at some term $(\p{p}_{i_l}, {i_l})$, the newly updated set $Q_{i_l}$ will bound the terms occurring at $(\p{p}_{i_{l+1}}, i_{l+1})$.

%  First we ignore all the terms $(Q,i)$ as in case (b) from all the terms needed to compute $G \cap \fk[x_1, \ldots, x_{k+1}]$, with all the remaining terms denoted by .
% The former case is when the corresponding part of delta set  $\Delta_{k+1, j}$ by Proposition \ref{prop: deltaset} is constructed before the first element in $G \cap \fk[x_1, \ldots, x_{k+1}]$ is found, and in that case $|\Delta_{k+1, j}|= |\Delta_k|$); while the latter

\subsection{Example}
Consider the ideal $I\subset \ff_{65521}[x_1, x_2]$ defined by its DRL \grobner basis ($x_1 < x_2$)
\begin{equation*}
\begin{split}
    G_1 = &\{x_2^4 + 2\,x_1^3 x_2 + 21\,x_2^3 + 11\,x_1 x_2^2 + 4\,x_1^2 x_2 + 22\,x_1^3 + 9\,x_2^2 + 17\,x_1 x_2 + 19\,x_1^2 + \\ &~2\,x_2 + 19\,x_1 + 5,  x_1^2 x_2^2 + 10\,x_2^3 + 12\,x_1^2 x_2 + 20\,x_1^3 + 21, ~x_1^4+ 15\,x_1^2 + 19\,x_1 + 3\}.
\end{split}
\end{equation*}
Here $\fnum_{65521}[x_1, x_2]/\ideal{G_1}$ is of dimension 12. Its basis, and further the multiplication matrices $T_1$ and $T_2$, can be computed accordingly.

Now we want to compute the \grobner basis $G_2$ of $I$ w.r.t.\ LEX. With a vector
\begin{equation*}
  \begin{split}
    \p{r} = (&6757, 43420, 39830, 45356, 52762, 17712, \\
    &27676,  17194, 138, 48036, 12649, 11037)^{t}\in \fnum_{65521}^{(12\times 1)}
  \end{split}
\end{equation*}
generated at random, the $2$-dimensional mapping $E$ is constructed. Then $\sakata()$ is applied term by term according to the LEX ordering, with $\Delta(\p{u})$ and the polynomial set $F$ valid up to $\p{u}$ shown in Table \ref{tab:sakata-exp}. For example, at the term $(4,0)$, the polynomial $x_1^2 + 62681\,x_1 + 41493 \in F$ is not valid up to $(4,0)$. Then the delta set is updated as $\{(0,0), (1,0), (2,0)\}$, and $F$ is reconstructed such that the new polynomial $x_1^3 + 62681\,x_1^2 + 35812\,x_1 + 18557$ is valid up to $(4,0)$.

The first polynomial in $G_2$:
$$g_1 =x_1^4+15\,x_1^2 + 19\,x_1 + 3$$
is obtained at the term $(7,0)$. Next $\sakata()$ is executed to compute other members of $I(E)$ according to the remaining term sequence $[x_2, x_1x_2, \ldots, $ $x_2^2, x_1^2x_2^2, \ldots,]$, until the other polynomial in $G_2$:
$$g_2=x_2^3+7\,x_1^2x_2^2 + 15\,x_1^2x_2 + 2\,x_1^3 + 9$$
is obtained at $(3,5)$. Now the main loop of Algorithm \ref{alg: bms} ends. Then one can easily verify that $\{g_1, g_2\}\subset G_2$ and $\dim(\fnum_{65521}[x_2, x_1]/\ideal{g_1, g_2})=12$, thus $G_2=\{g_1, g_2\}$.

\begin{table*}[h]
\centering
\resizebox{\textwidth}{!}{
\begin{tabular}{c c c}
\hline
Term $\p{u}$ & $\Delta(\p{u})$ & $F$: polynomial set valid up to $\p{u}$\\
\hline
$(0,0)$ & $(0,0)$&  $x_1, x_2$ \\
$(1,0)$ & -----& $x_1 + 65437, x_2$ \\
$(2,0)$ & $(0,0),(1,0)$& $x_1^2 + 65437\,x_1 +21672, x_2$\\
$(3,0)$ & -----& $x_1^2 + 62681\,x_1 + 41493, x_2$ \\
$(4,0)$ & $(0,0), (1,0), (2,0)$& $x_1^3 + 62681\,x_1^2 + 35812\,x_1 + 18557, x_2$ \\
$(5,0)$ & -----& $x_1^3 + 30688\,x_1^2 + 45566\,x_1 + 54643, x_2$ \\
$(6,0)$ & $(0,0), (1,0), (2,0), (3,0)$& $x_1^4 + 30688\,x_1^3 + 20026\,x_1^2 + 45766\,x_1 + 5434, x_2$ \\
$(7,0)$ & -----& $g_1, x_2$ \\
$(0,1)$ & -----& $g_1, x_2 + 65034\,x_1^3 + 24330\,x_1^2 + 14876\,x_1 + 52361$ \\
$(1,1)$ & -----& $g_1, x_2 + 64550\,x_1^3 + 37707\,x_1^2 + 48745\,x_1 + 7628$ \\
$(2,1)$ & -----& $g_1, x_2 + 38842\,x_1^3 + 5603\,x_1^2 + 45755\,x_1 + 44311$ \\
%$\vdots$ & $\vdots$ & $\vdots$\\
$(3,1)$ & -----& $g_1, x_2 + 9449\,x_1^3 + 20826\,x_1^2 + 39078\,x_1 + 38885$ \\
$(0,2)$ & $(0,0), (1,0), (2,0), (3,0), (0,1)$&
$\begin{array}{c}
    g_1, x_2^2  + 38885\,x_2 + 65360\,x_1^3 + 1782\,x_1^2 + 36000\,x_1 + 39469\\
    x_2x_1 + 20826\,x_1^3 + 28385\,x_1^2 + 55917\,x_1 + 37174
  \end{array}$\\
$\vdots$ & $\vdots$ & $\vdots$\\
\hline
\end{tabular}
}
\caption{Example for the {\sf BMS}-based method} \label{tab:sakata-exp}
\end{table*}

Here is an example where this method fails. Let $G=\{x_1^3, x_1^2x_2, $ $x_1x_2^2, x_2^3\}\subset \fnum_{65521}[x_1, x_2]$. Then the ideal $\ideal{G}$ is 0-dimensional with degree $D=6$. It is easy to see that $G$ is \grobner basis w.r.t.\ both DRL and LEX. Starting from $G$ as a \grobner basis w.r.t.\ DRL, the method based on the {\sf BMS} algorithm to compute the \grobner basis w.r.t.\ LEX will not be able to return the correct \grobner basis, even the base field itself is quite large and different random vectors $\p{r}$ are tried.

\section{Putting all methods together: top-level algorithm}\label{sec:main}

In this section, we combine the algorithms presented in the previous parts of this paper as the following integrated top-level algorithm, which performs the change of ordering of \grobner bases to LEX.

\begin{algorithm}[h]%\linesnumbered
  \KwIn{$G_1$, \grobner basis of a 0-dimensional ideal
    $I \subset \kx$ w.r.t.\ $<_1$}
  \KwOut{$G_2$, \grobner basis of $I$ or $\sqrt{I}$ w.r.t.\ LEX}

  $G_2 := \shapen(G_1, <_1)$\;
  \eIf{$G_2 \ne $ {\sf Fail}}
  {
    \Return{$G_2$}
  }
  {
    $G_2 := \shaped(G_1, <_1)$\;
    \eIf{$G_2 \ne$ {\sf Fail}}
    {
        \Return{$G_2$}
    }
    {
        $G_2 := \BMSbased(G_1, <_1)$\;
        \eIf{$G_2 \ne$ {\sf Fail}}
        {
            \Return{$G_2$}
        }
        {   \Return{$\fglm(G_1, <_1)$}
        }
    }
  }

\caption{Top-level algorithm $G_2 := \toplevel(G_1, <_1)$}\label{alg: main}
\end{algorithm}

We would like to mention that to integrate these three algorithms, one needs to skip some overlapped steps in the three algorithms, like computation of the canonical basis and the multiplication matrices, and the choice of random vectors, etc. If one does not seek for the \grobner basis of $\sqrt{I}$, that is to say, the multiplicities of the zeros are needed, then the deterministic invariant should be omitted.

Thanks to the feasibility in each algorithm to test whether the computed polynomial set is the correct \grobner basis, this top-level algorithm will automatically select which algorithm to use according to the input, until the original {\sf FGLM} one is called if all these algorithms fail. It is also a deterministic algorithm, though both the Wiedemann algorithm and the {\sf BMS}-based method will introduce randomness and probabilistic behaviors to the individual algorithms.

For an ideal in shape position, the probability for Algorithm \ref{alg: shapeN} to compute the correct \grobner basis is the same as that of computing the correct minimal polynomial in the Wiedemann algorithm for one choice of a random vector, which has been analyzed in \cite{W86S}. When Algorithm \ref{alg: shapeN} fails, the one based on the deterministic Wiedemann algorithm can tell us for sure whether the input ideal is in shape position, and return the \grobner basis of $\sqrt{I}$. However, the probability for the {\sf BMS}-based method to return the correct \grobner basis is still unknown.

\section{Multiplication matrix $T_1$: sparsity and complexity}\label{sec: random}

In the previous description and complexity analyses of all the algorithms, the multiplication matrices $T_1, \ldots, T_n$ are assumed known. In this section, for generic polynomial systems and the term ordering DRL, the multiplication matrix $T_1$ is exploited, on its sparsity and cost for construction. We are able to give an explicit formula to compute the number of dense columns in $T_1$, and we also analyze the asymptotic behavior of this number, which further leads to a finer complexity analysis for the change of ordering for generic systems. The term ordering is preassigned as DRL in this section without further notification.

%In fact in this paper we only deal with a sequence of $n$ polynomials of the same degree $d$, so we are interested in the generic properties of such a sequence, and we call it ``a generic sequence of degree $d$".

%Let $\mathcal{T}_{\leq d}$ and $\mathcal{T}_{d}$  be the sets of all terms of degree $\leq d$ and $d$ respectively. Then a $n$-variate random polynomial system of degree $d$, denoted by $R(n, d)$, is a set of $n$ polynomials such that each its polynomial $f$ is of the form
%    $$f = \sum_{\p{t}\in \mathcal{T}_{\leq d}} c_{\p{t}} \p{t},$$
%where all the coefficients $c_{\p{t}}$ are randomly chosen in $\fk$. Properties of a random polynomial system represent the generic properties of all systems of $n$ polynomials of degree $d$. More precisely, let $U$ be the set of all systems of $n$ polynomials of degree $d$, viewed as an affine space with the coefficients of the polynomials in the systems as the coordinates. Then a property of such systems is \emph{generic} if it holds on a Zariski-open in $U$.

\subsection{Construction of multiplication matrices}\label{subsec: construction-general}

Given the \grobner basis $G$ of a $0$-dimensional ideal $I$ w.r.t.\ DRL, let $B=[\epsilon_1, \ldots, \epsilon_D]$ be the canonical basis of $\kx / \ideal{G}$, and $L := \{\lt(g): g \in G\}$. The three cases of the multiplication $\epsilon_i x_j$ for the construction of the $i$th column of $T_j$ in {\sf FGLM} are reviewed below \cite{FGLM93E}.

\begin{enumerate}
  \item[(1)] The term $\epsilon_i x_j$ is in $B$: the coordinate vector of $\nf(\epsilon_i x_j)$ is $(0, \ldots, 0, 1, 0, \ldots, 0)^t$, where the position of $1$ is the same as that of $\epsilon_i x_j$ in $B$; \label{item: T1-1}
  \item[(2)] The term $\epsilon_i x_j$ is in $L$: the coordinate vector can be obtained easily from the polynomial $g\in G$ such that $\lt(g) = \epsilon_i x_j$; \label{item: T1-2}
  \item[(3)] Otherwise: the normal form of $\epsilon_i x_j$ w.r.t.\ $G$ has to be computed to get the coordinate vector. \label{item: T1-3}
\end{enumerate}

Obviously, the $i$th column of $T_j$ is sparse if case (a) occurs, thus a dense column can only come from cases (b) and (c). Furthermore, the construction for a column will not be free of arithmetic operations only if that column belongs to case (c). As a result, we are able to connect the cost for construction of the multiplication matrices with the numbers of dense columns in them.

%Let $B$ be the canonical basis of $\kx/\ideal{G_1}$. Then the set
%$$M = \{x_i b:\, b\in B, \, x_ib \not \in B~ (i=1, \ldots, n)\}$$
%is called the \emph{bordering} of $G_1$.
%
%\begin{proposition}
%  For each $m\in M$, either it is the leading term of a polynomial in $G_1$, or there exist $x_j$ and $m' \in M$ such that $m = x_j m'$.
%\end{proposition}

\begin{proposition}\label{prop: construction}
  Denote by $M_i$ the number of dense columns in the multiplication matrix $T_i$. Then the matrices $T_1, \ldots, T_n$ can be computed within $O(D^2\sum_{i=1}^n M_i)$ arithmetic operations.
\end{proposition}

\begin{proof}
  Direct result from the proof of Proposition 3.1 in \cite{FGLM93E}.
\end{proof}

As shown in Section \ref{sec:shape},  among all multiplication matrices, $T_1$ is the most important one, and it is also of our main interest. However, for an arbitrary ideal, now we are not able to analyze the cost to construct $T_1$ by isolating it from the others in Proposition \ref{prop: construction}, for the analysis on $T_1$ needs information from the other matrices too.

In the following parts we first focus on generic sequences which impose stronger conditions on $T_1$ so that the analyses on it become feasible. We show that the construction of $T_1$ for generic sequences is free and present finer complexity results based on an asymptotic analysis. %Then we return to the non-generic case with the discussion on a useful method based on the change of coordinates to avoid the cost for construction of $T_1$, and further an improved strategy for solving zero-dimensional polynomial systems.

\subsection{Generic sequences and Moreno-Soc\'{\i}as conjecture}

Let $P = [f_1, \ldots, f_n]$ be a sequence of polynomials in $\kx$ of degree $d_1, \ldots, d_n$. If $d_1 = \cdots =d_n = d$, we call it a sequence of degree $d$. We are interested in the properties of the multiplication matrices for the ideal generated by $P$ if $f_1, \ldots, f_n$ are chosen ``at random". Such properties can be regarded generic in all sequences. More precisely, let $U$ be the set of all sequences of $n$ polynomials of degree $d_1, \ldots, d_n$, viewed as an affine space with the coefficients of the polynomials in the sequences as the coordinates. Then a property of such sequences is \emph{generic} if it holds on a Zariski-open in $U$. Next for simplicity, we will say some property holds ``for a generic sequence" if it is a generic one, and also $P$ is ``a generic sequence" if its properties of our interest are generic.

For a generic sequence $[f_1, \ldots, f_n]$, its properties concerning the \grobner basis computation, in particular the canonical basis, are the same as $[f_1^h, \ldots, f_n^h]$, where $f_i^h$ is the homogeneous part of $f_i$ of the highest degree. That is to say, we only need to study homogeneous generic sequences, which are also those studied in the literature. Hence in the following part of this section, a generic sequence is further assumed homogeneous.

Since we restrict to the situation where the number of polynomials is equal to that of variables, a generic sequence is a \emph{regular} one \cite{L83G}. We first recall the well-known characterization of a regular sequence via its Hilbert series.

\begin{theorem}\label{thm: hilbertS}
   Let $[f_1, \ldots, f_r]$ be a sequence in $\kx$ with $\deg(f_i) = d_i$. Then it is regular if and only if its Hilbert series is
        $$\frac{\prod_{i=1}^r (1-z^{d_i})}{(1-z)^n}.$$
\end{theorem}

Let $P$ be a generic sequence of degree $d$. Then we know its Hilbert series is
\begin{equation}\label{eq: hilbert}
  H(n, d) := (1-z^d)^n / (1-z)^n = (1+z+z^2+\cdots + z^{d-1})^n,
\end{equation}
from which one can easily derive that the degree of $\ideal{P}$ is $d^n$, and that the greatest total degree of terms in the canonical basis is $(d-1)n$.

\grobner bases of generic sequences w.r.t.\ DRL have been studied in \cite{M2003d}.  A term ideal $J$ is said to be a \emph{weakly reverse lexicographic ideal} if the following condition holds: if $\p{t} \in J$ is a minimal generator of $J$, then $J$ contains every term of the same total degree as $\p{t}$ which is greater than $\p{t}$ w.r.t.\ some term ordering. For the term ordering DRL, we have the following conjecture due to Moreno-Soc\'{\i}as.

~\newline
{\bf Moreno-Soc\'{\i}as conjecture} (\cite{M1991a})~
\emph{Let $\fk$ be an infinite field, and $P = [f_1, \ldots, f_n]$ a generic sequence in $\kx$ with $\deg(f_i) = d_i$. Then $\lt(\ideal{P})$, the leading term ideal of $\ideal{P}$ w.r.t.\ DRL, is a weakly reverse lexicographical ideal.}
~\newline

The Moreno-Soc\'{\i}as conjecture is proven true for the codimension $2$ case and for some special ideals for the codimension $3$ case \cite{AJL2001g, C2007g}. It has been proven that this conjecture implies the Fr\"{o}berg conjecture on the Hilbert series of a generic sequence, which is well-known and widely acknowledged true in the symbolic computation community \cite{P2010g}. %Conditions for term ideals to be weakly reverse lexicographical are shown in \cite{CP2008c}.

\begin{proposition}\label{prop: conjecture-prop}
  Use the same notations as those in the Moreno-Soc\'{\i}as conjecture. If this conjecture holds, then for
  a term $\p{u}\in \lt(\ideal{P})$, any term $\p{v}$ such that $\deg(\p{u}) = \deg(\p{v})$ and $\p{v} > \p{u}$ is also in $\lt(\ideal{P})$.
\end{proposition}

\begin{proof}
  If $\p{u}$ is a minimal generator of $\lt(\ideal{P})$, then the conclusion is a direct result from the Moreno-Soc\'{\i}as conjecture. Else there exists one minimal generator  $\tilde{\p{u}} \neq \p{u}$ such that $\p{u} \succ \tilde{\p{u}}$. As for any $\p{w}$ such that $\deg(\p{w})=\deg(\tilde{\p{u}})$ and $\p{w} > \tilde{\p{u}}$, we know $\p{w}\in \lt(\ideal{P})$. Then we can always find a term $\tilde{\p{v}}\in \lt(\ideal{P})$ such that $\p{v}\succ \tilde{\p{v}}$. For example, construct $\tilde{\p{v}} = \tilde{\p{u}} - \p{u} + \p{v}$. If $\tilde{\p{v}}$ is a term, then it suffices; otherwise the biggest term $\overline{\p{v}}$ such that $\deg(\overline{\p{v}}) = \deg(\tilde{\p{u}})$ will work. This ends the proof.
\end{proof}

As Proposition \ref{prop: conjecture-prop} implies, the Moreno-Soc\'{\i}as conjecture imposes a stronger requirement on the structure of the terms in $\lt(\ideal{P})$ for a generic sequence $P$. For the bivariate case, once a term $\p{u}$ is known to be an element in $\lt(\ideal{P})$, the terms in $\lt(\ideal{P})$ determined by it are illustrated in Figure \ref{fig: conjecturenew} (left), and furthermore, in the right figure the shape all terms in $\lt(\ideal{P})$ form.

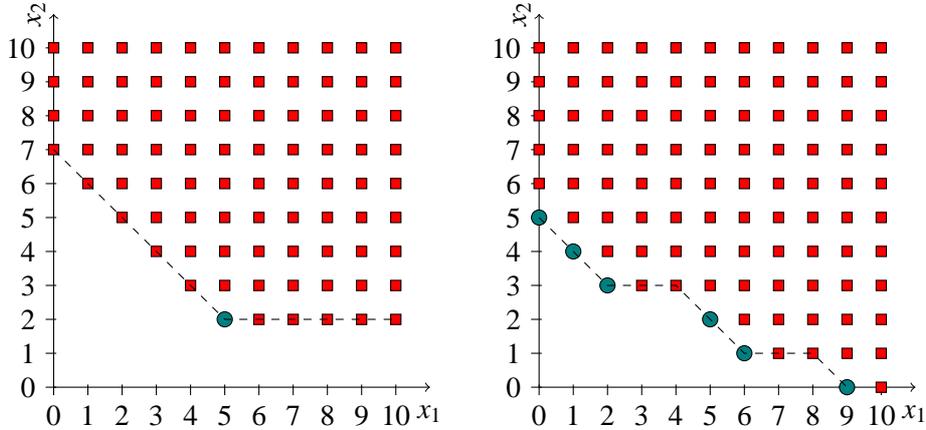
\begin{figure}[h]
\centering
\begin{tikzpicture}[y=0.45cm, x=0.45cm]
        %axis
        \draw[->] (0,0) --  (11,0);
        \draw[->] (0,0) -- (0,11);
        %ticks
        \foreach \x in {0,...,10}
                \draw (\x,1pt) -- (\x,-3pt)
                        node[anchor=north] {\x};
        \foreach \y in {0,...,10}
                \draw (1pt,\y) -- (-3pt,\y)
                        node[anchor=east] {\y};
        %labels
        \node[below=0.1cm] at (11,0 ) {\(x_{1}\)};
        \node[rotate=90] at (-0.5,11) {\(x_{2}\)};
        %staircase
        \foreach \y in {7,...,10}\path[draw=black,fill=red] (0-0.15,\y-0.15) rectangle +(0.3,0.3);
        \foreach \y in {6,...,10}\path[draw=black,fill=red] (1-0.15,\y-0.15) rectangle +(0.3,0.3);
        \foreach \y in {5,...,10}\path[draw=black,fill=red] (2-0.15,\y-0.15) rectangle +(0.3,0.3);
        \foreach \y in {4,...,10}\path[draw=black,fill=red] (3-0.15,\y-0.15) rectangle +(0.3,0.3);
        \foreach \y in {3,...,10}\path[draw=black,fill=red] (4-0.15,\y-0.15) rectangle +(0.3,0.3);
        \foreach \y in {3,...,10}\path[draw=black,fill=red] (5-0.15,\y-0.15) rectangle +(0.3,0.3);
         \path[draw=black,fill=green!50!blue] (5,2) circle (0.1cm);
        \foreach \x in {6,...,10}
           \foreach \y in {2,...,10}\path[draw=black,fill=red] (\x-0.15,\y-0.15) rectangle +(0.3,0.3);
 \draw[dashed] (0,7) -- (5,2) -- (10,2);
\end{tikzpicture}
~~
\begin{tikzpicture}[y=0.45cm, x=0.45cm]
        %axis
        \draw[->] (0,0) --  (11,0);
        \draw[->] (0,0) -- (0,11);
        %ticks
        \foreach \x in {0,...,10}
                \draw (\x,1pt) -- (\x,-3pt)
                        node[anchor=north] {\x};
        \foreach \y in {0,...,10}
                \draw (1pt,\y) -- (-3pt,\y)
                        node[anchor=east] {\y};
        %labels
        \node[below=0.1cm] at (11,0 ) {\(x_{1}\)};
        \node[rotate=90] at (-0.5,11) {\(x_{2}\)};
        %staircase
        \foreach \y in {6,...,10}\path[draw=black,fill=red] (0-0.15,\y-0.15) rectangle +(0.3,0.3);
        \foreach \y in {5,...,10}\path[draw=black,fill=red] (1-0.15,\y-0.15) rectangle +(0.3,0.3);
        \foreach \y in {4,...,10}\path[draw=black,fill=red] (2-0.15,\y-0.15) rectangle +(0.3,0.3);
        \foreach \y in {3,...,10}\path[draw=black,fill=red] (3-0.15,\y-0.15) rectangle +(0.3,0.3);
        \foreach \y in {3,...,10}\path[draw=black,fill=red] (4-0.15,\y-0.15) rectangle +(0.3,0.3);
        \foreach \y in {3,...,10}\path[draw=black,fill=red] (5-0.15,\y-0.15) rectangle +(0.3,0.3);
        \foreach \x in {7,...,10}\path[draw=black,fill=red] (\x-0.15,1-0.15) rectangle +(0.3,0.3);
         \path[draw=black,fill=red] (10-0.15,0-0.15) rectangle +(0.3,0.3);
         \path[draw=black,fill=green!50!blue] (0,5) circle (0.1cm);
         \path[draw=black,fill=green!50!blue] (1,4) circle (0.1cm);
         \path[draw=black,fill=green!50!blue] (2,3) circle (0.1cm);
         \path[draw=black,fill=green!50!blue] (5,2) circle (0.1cm);
         \path[draw=black,fill=green!50!blue] (6,1) circle (0.1cm);
         \path[draw=black,fill=green!50!blue] (9,0) circle (0.1cm);
        \foreach \x in {6,...,10}
           \foreach \y in {2,...,10}\path[draw=black,fill=red] (\x-0.15,\y-0.15) rectangle +(0.3,0.3);
 \draw[dashed] (0,5) -- (2,3) -- (4,3) -- (6,1) -- (8,1) -- (9,0);
\end{tikzpicture}
\caption{One term $\p{u}\in \lt(\ideal{P})$ (\protect\tikz{\protect \path[draw=black,fill=green!50!blue] (0,0) circle (0.1cm);}) and the terms it determines (\protect\tikz{\protect \path[draw=black,fill=red] (0,0) rectangle +(0.2,0.2);}) ~/~ Minimal generators of $\lt(\ideal{P})$ (\protect\tikz{\protect \path[draw=black,fill=green!50!blue] (0,0) circle (0.1cm);}) and terms in $\lt(\ideal{P})$ (\protect\tikz{\protect \path[draw=black,fill=red] (0,0) rectangle +(0.2,0.2);})}\label{fig: conjecturenew}

\end{figure}

The base field in the Moreno-Soc\'{\i}as conjecture is restricted infinite. According to our preliminary experiments on randomly generated sequences over fields of large cardinality, we find no counterexample of this conjecture. As a result, we will consider it true and use it directly. The following variant of Moreno-Soc\'{\i}as conjecture, which is more convenient in our setting, can be derived easily from Proposition \ref{prop: conjecture-prop}. %, so in the following part of this section, we assume the field $\fk$ is infinite.

%\begin{conjecture}\label{thm: smallest}
~\newline
{\bf Variant of  Moreno-Soc\'{\i}as conjecture}~
  \emph{Let $\fk$ be an infinite field, $P\in \kx$ a generic sequence of degree $d$, and $B$ the canonical basis of $\kx/\ideal{P}$ w.r.t.\ DRL. Denote by $B(k)$ the set of terms of total degree $k$ in $B$. Then for $k=1, \ldots, (d-1)n$, $B(k)$ consists of the first $|B(k)|$ smallest terms in all terms of total degree $k$.}
~\newline%\end{conjecture}

\subsection{Sparsity and construction}\label{subsec: sparsity1}

Let $P\subset \kx$ be a generic sequence, and $G$ the \grobner basis of $\ideal{P}$. Then polynomials in $G$ can be assumed dense (in the case when $\fk$ is of characteristic $0$ or of large cardinality). As the number of dense columns in $T_1$ will directly lead to a bound on the number of nonzero entries in $T_1$, the study of $T_1$ sparsity is reduced to that of how many cases of (2) and (3) happen. Combining the Hilbert series of a generic sequence and our variant of Moreno-Soc\'{\i}as conjecture, we are able to give the counting of the dense columns in $T_1$.

\begin{proposition}\label{thm: dense}
    Let $\fk$, $P$, $B$ and $B(k)$ be the same as those in the Moreno-Soc\'{\i}as conjecture variant. If the Moreno-Soc\'{\i}as conjecture holds, then the number of dense columns in the multiplication matrix $T_1$ is equal to the greatest coefficient in the expansion of $(1+z+\cdots+z^{(d-1)})^n$.
\end{proposition}

\begin{proof}
    Let $k' = (d-1)n$, and denote by $\term{T}^{(k)}$ be set of all terms in $\kx$ of total degree $k$.

    Suppose that $\p{u}$ is the $l$th smallest term in $\mathcal{T}^{(k)}$. Then $x_1\p{u}$ is still the $l$th smallest term in $\mathcal{T}^{(k+1)}$. Hence from the conjecture variant, if $|B(k)| \leq |B(k+1)|$, then for every $\p{u}\in B(k)$, $x_1\p{u}$ is still in $B(k+1)$. Therefore it belongs to case (1) we reviewed in Section \ref{subsec: sparsity1}, and the corresponding column in $T_1$ is a sparse one. If $|B(k)| > |B(k+1)|$, we will have $|B(k)| - |B(k+1)|$ dense columns which come from the fact that they belong to case (2) or (3).

    As the coefficients in the the expansion of $(1+z+\cdots+z^{(d-1)})^n$ are symmetric to the central coefficient (or the central two when $(d-1)n$ is odd), the condition $|B(k)| > |B(k+1)|$ holds for the first time when $k = k_0$, the index of the central term (or of the second one in the central two terms). Then the number of dense columns is
      \begin{equation*}
        \begin{split}
          &(|B(k_0)| - |B(k_0 + 1)|) + (|B(k_0+1)| - |B(k_0 + 2)|) \\
          &+ \cdots + (|B(k'-1)| - |B(k')|) + |B(k')| =  |B(k_0)|.
        \end{split}
      \end{equation*}

      That ends the proof, for such a coefficient $|B(k_0)|$ is exactly the greatest one.
\end{proof}

The Hilbert series is usually used to analyze the behaviors of \grobner basis computation, for example the regularity of the input ideal. As the leading terms of polynomials in the \grobner basis and the canonical basis determine each other completely, it is also natural to have Proposition \ref{thm: dense}, which links the canonical basis and Hilbert series.

\begin{remarks}
    When $d=2$, the number of dense columns in $T_1$ is the binomial coefficient $C_n^{k_0}$, where
    \begin{equation*}
      k_0 = \left\{
    \begin{array}{cl}
        n/2 & \mbox{\quad if $n$ is even;} \\
        \frac{n+1}{2} & \mbox{\quad if $n$ is odd.}
    \end{array}
\right.
    \end{equation*}
    For the case $d=3$, such the greatest coefficient is called the \emph{central trinomial coefficient}.
\end{remarks}

\begin{corollary}\label{cor: sparsity}
  If the Moreno-Soc\'{\i}as conjecture holds, then the percentage of nonzero entries in $T_1$ for a generic sequence of degree $d$ is bounded by $(m_0+1)/D$, where $m_0$ is the number of dense columns computed from Proposition \ref{thm: dense}.
\end{corollary}

\begin{proof}
  The number of nonzero entries in the dense columns is bounded by $m_0 D$, and that in the other columns is smaller than $D$.
\end{proof}

Assuming the correctness of the Moreno-Soc\'{\i}as conjecture, we can take a step forward from Proposition \ref{thm: dense}. That is, we show case (3) will not occur during the construction of $T_1$.

\begin{proposition}
  Follow the notations in the Moreno-Soc\'{\i}as conjecture. If the conjecture holds, then for any term $\p{u} \not \in \lt(\ideal{P})$, $x_1 \p{u}$ is either not in $\lt(\ideal{P})$ or a minimal generator of $\lt(\ideal{P})$.
\end{proposition}

\begin{proof}
  Suppose $x_1\p{u} = (u_1, \ldots, u_n) \in \lt(\ideal{P})$ is not a minimal generator. We will draw a contradiction by showing $\p{u} \in \lt(\ideal{P})$ under such an assumption.

  Without loss of generality, we can assume each $u_i \neq 0$ for $i=1,\ldots, n$, otherwise we can reduce to the $n-1$ case by ignoring the $i$th component of $\p{u}$. As $x_1\p{u}$ is not a minimal generator, there exist a $k~(1\leq k \leq n)$ such that $\p{u}^{(k)} := (u_1, \ldots, u_k -1, \ldots, u_n)$ is in $\lt(\ideal{P})$. The case when $k=1$ is trivial. Otherwise, since $\deg(\p{u}^{(k)}) = \deg(\p{u})$ and $\p{u}^{(k)} < \p{u}$, by Proposition \ref{prop: conjecture-prop} we know $\p{u} \in \lt(\ideal{P})$.
\end{proof}

\begin{corollary}\label{cor: case2only}
  If the Moreno-Soc\'{\i}as conjecture holds, then the number of dense columns in $T_1$ for generic sequences is equal to the cardinality of $\{g \in G_1:\, x_1 | \lt(g)\}$, where $G_1$ is the \grobner basis w.r.t.\ DRL.
\end{corollary}

\begin{remarks}
  By Corollary \ref{cor: case2only}, for generic sequences, to construct $T_1$ one only needs to find the leading term of which polynomial in $G_1$ is a given term $x_1 \p{u}$. Thus we can conclude that the construction of $T_1$ is free of arithmetic operations. Even for real implementations, the cost for constructing $T_1$ is also quite small compared with that for the change of ordering (see Section \ref{sec:exp} for the timings). Bearing in mind that the ideal generated by a generic sequence is in shape position, we know the complexity in Theorem \ref{prop:complexity shape} is indeed the complete complexity for the change of ordering for generic sequences, including construction of $T_1$, the only multiplication matrix  needed.
\end{remarks}

%With the above results, one can further obtain the complexity for computing all multiplication matrices.

\subsection{Asymptotic analysis}
Next we study the asymptotic behavior of the number of dense columns in $T_1$ for a generic sequence of degree $d$, with $n$ fixed and $d$ increasing to $+\infty$. These results are mainly derived from a more detailed asymptotic analysis of coefficients of the Hilbert series of semi-regular systems in \cite{B2004e, BFS04}, where standard methods in asymptotic analysis, like the saddle-point and coalescent saddle points methods, are applied.

The target of this subsection is to find the dominant term of the greatest coefficient in the expansion of the Hilbert series $H(n,d)$ in \eqref{eq: hilbert}, as $d$ tends to $+\infty$ and $n$ is fixed. First one writes the $m$th coefficient $I_d(m)$ of $H(n,d)$ with the Cauchy integration:
    $$I_d(m) = \frac{1}{2\pi \im}\oint \frac{H(n,d)(z)}{z^{m+1}} dz = \frac{1}{2\pi \im}\oint \frac{(1-z^d)^n}{(1-z)^n z^{m+1}} dz.$$

With $F(z) := \frac{(1-z^d)^n}{(1-z)^n z^{m+1}} = e^{f(z)}$ and $g(z) :=1$, $I_d(m)$ becomes the form convenient to the asymptotic analysis
$$I_d(m) =  \frac{1}{2\pi \im}\oint g(z) e^{f(z)} dz.$$
Suppose the greatest coefficient in $H(n,d)$ comes from the $m_0$th term. Since we are interested in the asymptotic behavior, we can assume $m_0 = (d-1)n/2$. As a special case of \cite[Lemma 4.3.1]{B2004e}, we have the following result.

\begin{proposition}
Suppose $m_0 = (d-1)n/2$. Then the dominant term of $I_d(m_0)$ is
$$I_d(m_0) \sim \sqrt{\frac{1}{2\pi f''(r_0)}}  e^{f(r_0)},$$
where $f(z) = n\log(\frac{1-z^d}{1-z}) - (m+1)\log(z)$, and $r_0$ is the positive real root of $f'(z)$. Furthermore, $r_0$ tends to $1$ as $d$ increases to $+\infty$.
\end{proposition}

To prove the fact that the positive real root $r_0$ of $f'(z)$ tends to $1$, one needs to use the equality $m_0 = (d-1)n/2$. Other parts of the proof are the same as those in \cite[Section 4.3.2]{B2004e}.

Next we investigate the value of $f''(r_0)$ and $F(r_0)$ in the dominant part of $I_d(m_0)$ as $r_0$ tends to $1$. Set $h(z) := \frac{1-z^d}{1-z} = 1+z+\cdots +z^{d-1}$. Then
$$f''(z) = n \frac{h''(z) h(z) - h'(z)^2}{h(z)^2} + \frac{d+1}{z}.$$
Noting that $h(1) = d$, $h'(1) = d(d-1)/2$, and $h''(1) = d(d-1)(d-2)/3$, we have
$$f''(1) = nd^2/12 + O(d).$$
With the easily obtained equality $F(1)=d^n$, we have the following asymptotic estimation of $I_d(m_0)$.

\begin{corollary}\label{prop: m0}
  Let $n$ be fixed. As $d$ tends to $+\infty$, $I_d(m_0) \sim \sqrt{\frac{6}{n\pi}}d^{n-1}$.
\end{corollary}

This asymptotic estimation of the greatest coefficient in $H(n,d)$ accords with the theoretical one. Figure \ref{fig: asymp} shows the number of dense columns derived from both Proposition \ref{thm: dense} and Corollary \ref{prop: m0}. As can be shown from this figure, the asymptotic estimation is good, even when $d$ is small.

\begin{figure}[h]
  \centering
  \includegraphics[width=6cm]{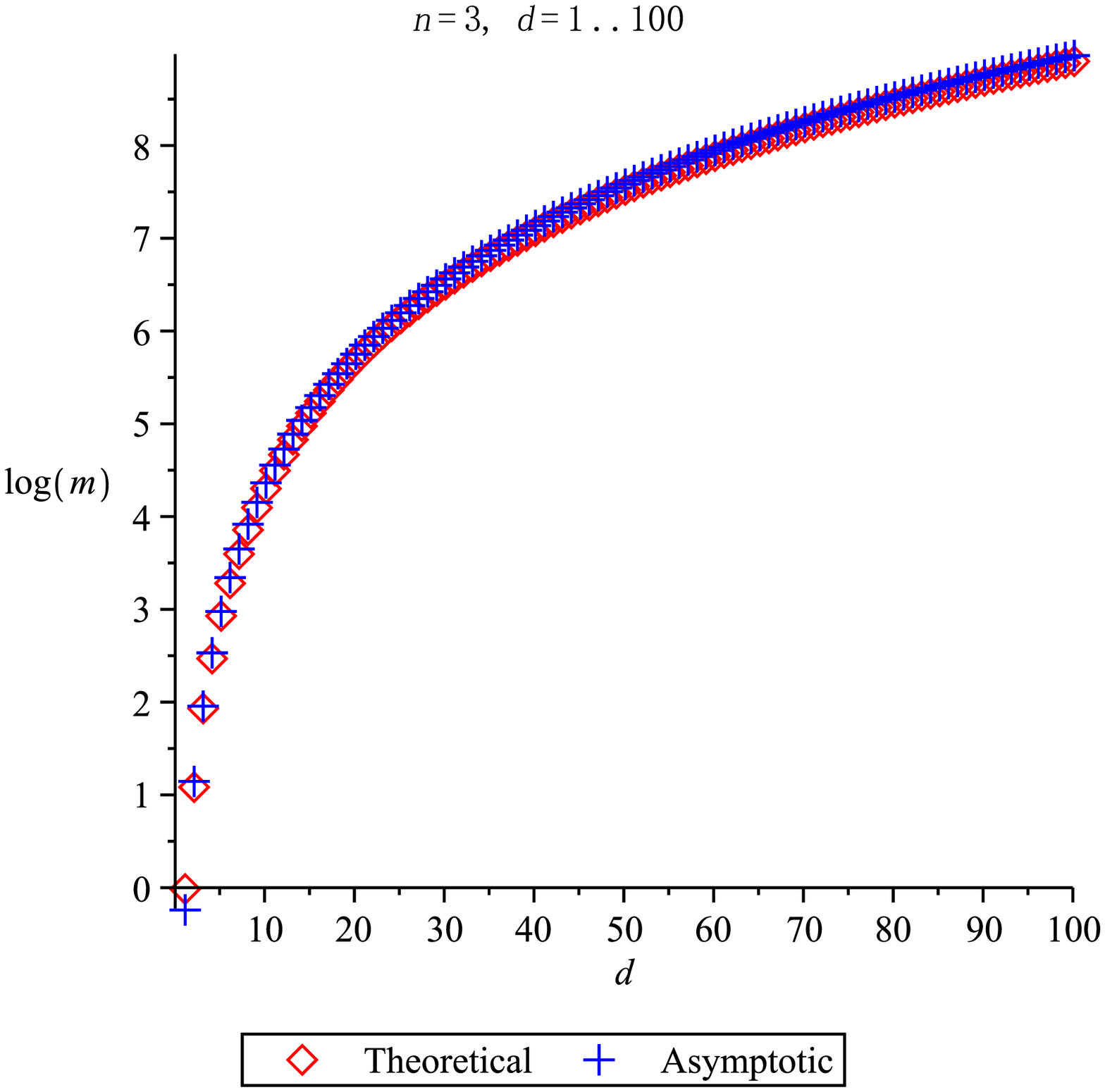}
  \hspace{-0.2cm}%
  \includegraphics[width=6cm]{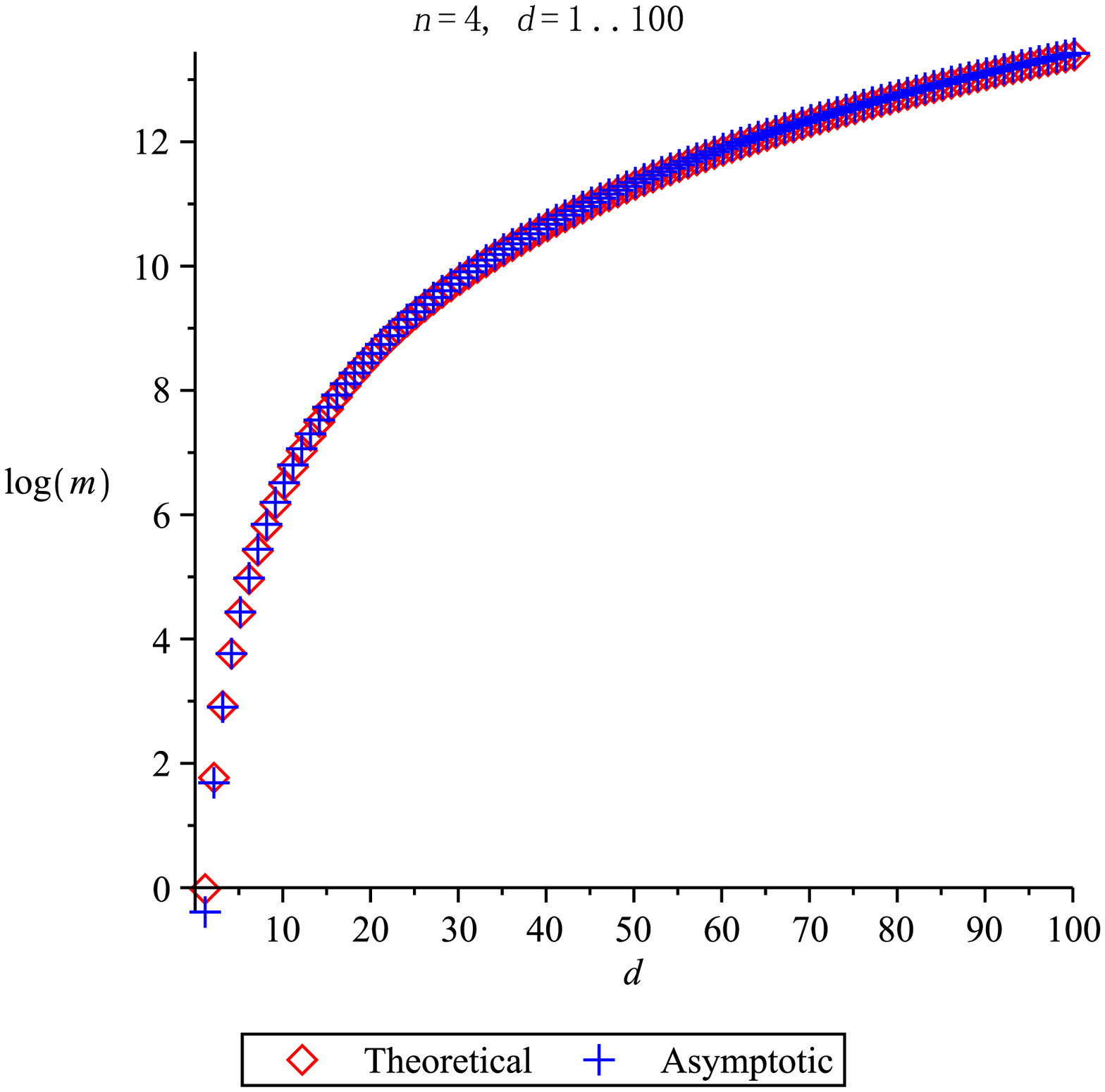}
  %\hspace{0.1cm}%
  %\includegraphics[width=6cm]{pic/n5-d100.eps}
%  \hspace{0.1cm}%
%  \includegraphics[width=6cm]{pic/n6-d100.eps}
  \caption{Number of dense columns in $T_1$ for $n=3, 4$ and $d=1, \ldots, 100$}\label{fig: asymp}
\end{figure}

\begin{corollary}\label{cor: random-complexity}
  Let $n$ be fixed. As $d$ tends to $+\infty$, if the Moreno-Soc\'{\i}as conjecture holds, then the following statements hold:
  \begin{enumerate}
    \item the percentage of nonzero entries in $T_1$ is $\sim \sqrt{\frac{6}{n\pi}}/d$;
    \item for a generic sequence of degree $d$, the complexity in Theorem \ref{prop:complexity shape} is $O(\sqrt{\frac{6}{n\pi}} D^{2+\frac{n-1}{n}})$.
  \end{enumerate}
\end{corollary}

  As Corollary \ref{cor: random-complexity} shows, for a generic sequence, the multiplication matrix $T_1$ become sparser as $d$ increases. Furthermore, the complexity of Algorithm \ref{alg: shapeN} is smaller in both the exponent and constant compared with {\sf FGLM}.

\begin{remarks}
  Here we only consider the case when $n$ is fixed and $d$ tends to $+\infty$, while the asymptotic behaviors of the dual case when $d$ is fixed and $n$ tends to $+\infty$ have been studied in \cite{BFS04} for the special value $d=2$.
\end{remarks}

\section{Experiments} \label{sec:exp}

The first method for the shape position case, namely Algorithm \ref{alg: shapeN}, has been implemented in C over fields of characteristic $0$ and finite fields. A preliminary implementation of the {\sf BMS}-based method for the general case has been done in {\sf Magma} over large finite fields. Benchmarks are used to test the correctness and efficiency of these two methods. All the experiments were made under Scientific Linux OS release 5.5 on 8 Intel(R) Xeon(R) CPUs E5420 at 2.50 GHz with 20.55G RAM.

Table \ref{tab:time} records the timings (in seconds) of our implementations of $F_5$ and Algorithm \ref{alg: shapeN} applied to benchmarks including theoretical ones like Katsura systems (Katsura$n$) and randomly generated quadratic polynomial systems of $n$ variables (Random$n$), and practical ones like MinRank problems from Cryptography~\cite{FSS10} and algebraic cryptanalysis of some curve-based cryptosystem (Edwards). In this table, $D$ denotes the degree of the input ideal, and the column "Density" means the percentage of nonzero entries in the multiplication matrix $T_1$. The instances marked with \textdagger~ are indeed not in shape position, and the timings for such instances only indicate those of computing the univariate polynomial in the LEX \grobner basis. The performances of the DRL \grobner basis computation and {\sf FGLM} in {\sf Magma} (version 2-17-1) and {\sf Singular} (version 3-1-2), together with the speedup factors of our implementation for the change of ordering, are also provided.

As shown by this table, the current implementation of Algorithm \ref{alg: shapeN} outperforms the {\sf FGLM} implementations in {\sf Magma} and {\sf Singular}. Take the Random13 instance for example, the {\sf FGLM} implementations in {\sf Magma} and {\sf Singular} take $10757.4$ and $19820.2$ seconds respectively, while the new implementation only needs $193.5$ seconds. This is around $54$ and $101$ times faster. Such an efficient implementation is now able to manipulate ideals in shape position of degree greater than $40000$. It is also important to note that with this new algorithm, the time devoted to the change of ordering is somehow of the same order of magnitude as the DRL Gr\"obner basis computation.

\begin{table*}[ht]
\centering
\label{tab:time}
\resizebox{\textwidth}{!}{
\begin{tabular}{c|cc|cccccc|cc}
\multicolumn{3}{c}{ } & \multicolumn{2}{c}{FGb}& \multicolumn{2}{c}{Magma}&\multicolumn{2}{c}{Singular}&\multicolumn{2}{c}{Speedup}\\
\hline
Name & $D$ & Density & $F_5$(C) &\begin{small}New Algorithm\end{small} & $F_4$ & FGLM&DRL &FGLM &Magma &Singular\\
\hline
%Katsura10 & $2^{10}$  &23.30\%&0.9  &0.5  &2.8  &22.2  &81.8  &42.7  &41.0&79.1 \\
Katsura11  & $2^{11}$ &21.53\%&4.9  &3.4  &18.2  &178.6  &632.0  &328.4  &52.7&96.9\\
Katsura12 & $2^{12}$ & 21.26\%&31.9  &26.3  &147.9  & 1408.1  &5061.8  &2623.5  &53.6&99.8\\
Katsura13 & $2^{13}$ &19.86\%&186.3  &189.1  &1037.2  &10895.4  &&&57.6&\\
Katsura14 & $2^{14}$ &19.64\%&1838.9  &1487.4  &9599.0  &87131.9  &&&58.5&\\
Katsura15 & $2^{15}$ &18.52\%&11456.3  &12109.2  &&&&&&\\
\hline
Random 11&$2^{11}$&21.53\%&4.7  &3.4  &18.1  &169.3  &623.9  &328.6  &49.2&95.5\\
Random 12&$2^{12}$&21.26\%&26.6  &26.9  &134.9  &1335.8  &4867.4  &2581.1  &49.6&95.8\\
Random 13&$2^{13}$&19.98\%&146.8  &193.5  &949.6  &10757.4  &36727.0  &19820.2&55.6&102.4\\
Random 14&$2^{14}$&19.64\%&1000.7  &1489.5  &7832.4  &84374.6&&&56.6&\\
Random 15&$2^{15}$&18.52\%&6882.5  &10914.02  &&&&&&\\
\hline
MinR(9,6,3) & 980&  26.82\%&1.1  &0.5  &6.3  &22.7  &137.5  &38.1  &43.6&73.2\\
MinR(9,7,4) & 4116& 22.95\%&28.4  &28.5  &208.1  &1360.4  &4985.8  &2490.3  &47.7&87.4\\
MinR(9,8,5) & 14112&19.04\%&543.6  &1032.8  &&&&&&\\
MinR(9,9,6) & 41580 &16.91\%&9048.2  &22171.3  &&&&&&\\
\hline
Weierstrass & 4096& 7.54\%&4.0  &9.0  &5.8  &418.3  &72.4  &1823.6  &46.7&203.7\\
Edwards \textdagger  & 4096 &3.41\%&0.1  &2.4  &0.2  &176.7  &1.0  &839.9  &72.7&345.6\\
Cyclic 10 \textdagger  &    34940 &1.00\%&&3586.9  &
\multicolumn{2}{r}{\(>\){\small 16 hrs}
{\small and} \(>\){\small 16 Gig}}&&&&\\
\hline
\end{tabular}}
\caption{Timings of the method for the shape position case from DRL to LEX}
\end{table*}

Table \ref{tab:sakata-time} illustrates the performances of Algorithm \ref{alg: bmsbased} for the general case. As currently this method is only implemented preliminarily in {\sf Magma}, only the number of field multiplications and other critical parameters are recorded, instead of the timings.

Benchmarks derived from Cyclic 5 and 6 instances are used. Instances with ideals in shape position (marked with \textdaggerdbl) are also tested to demonstrate the generality of this method. Besides $n$ and $D$ denoting the number of variables and degree of the input ideal, the columns ``Mat Density" and ``Poly Density" denote the maximal percentage of nonzero entries in the matrices $T_1, \ldots, T_n$ and the density of resulting \grobner bases respectively. The following 4 columns record the numbers of passes in the main loop of Algorithm \ref{alg: bmsbased}, matrix multiplications, reductions and field multiplications.

As one can see from this table, the numbers of passes accord with the bound derived in theorem \ref{thm:loop}, and the number of operations is less than the original {\sf FGLM} algorithm for Cyclic-like benchmarks. However, for instances of ideals in shape position, this method works but the complexity is not satisfactory. This is mainly because the resulting \grobner bases in these cases are no longer sparse, and thus the reduction step becomes complex. Fortunately, in the top-level algorithm \ref{alg: main}, it is not common to handle such ideals in shape position with this method.

\begin{table}[h]
 \centering
 \resizebox{\textwidth}{!}{
\begin{tabular}{c |c@{} c @{}c@{} c| c @{}c @{}c @{}c}
\hline
Name & ~~$n$~~ & ~~$D$~~ & ~Mat Density~ & Poly Density & ~\#\,Passes~ & ~\#\,Mat.~ & ~\#\,Red.~ & ~\#\,Multi.~\\
\hline
Cyclic5-2 & $2$  &  $55$  & $4.89\%$ &  $17.86\%$   & 165 & 318 & 107 & $nD^{2.544}$\\
Cyclic5-3 & $3$  &  $65$  & $8.73\%$ &  $19.7\%$   & 294 & 704 & 227 & $nD^{2.674}$\\
Cyclic5-4 & $4$  &  $70$  & $10.71\%$ &  $21.13\%$   & 429 & 1205 & 355 & $nD^{2.723}$\\
Cyclic5& $5$  &  $70$  & $12.02\%$ &  $21.13\%$   & 499 & 1347 & 421 & $nD^{2.702}$\\
Cyclic6 & $6$  &  $156$  & $11.46\%$ &  $17.2\%$   & 1363 & 4464 & 1187 & $nD^{2.781}$\\
\hline
Uteshev Bikker \textdaggerdbl & $4$  &  $36$  & $60.65\%$ &  $100\%$   & 179 & 199 & 105 & $nD^{2.992}$\\
D1 \textdaggerdbl  & $12$ &  $48$  & $34.2\%$ &  $51.02\%$   & 624 & 780 & 517 & $nD^{2.874}$\\
Dessin2-6 \textdaggerdbl & $6$  &  $42$  & $46.94\%$ &  $100\%$   & 294 & 336 & 205 & $nD^{2.968}$\\
\hline
\end{tabular}
}
\caption{Performances of Algorithm \ref{alg: bmsbased} for the general case from DRL to LEX}\label{tab:sakata-time}
\end{table}

%\section{Concluding Remarks} \label{sec:remark}
%
%All algorithms proposed here follow the idea of the Wiedemann algorithm. To be precise, we first construct the linearly recurring sequence or relation with the multiplication matrices in the {\sf FGLM} algorithm, and find its minimal generator(s) with  the Berlekamp--Massey algorithm or its multi-dimensional generalization {\sf BMS} algorithm. This is how the sparsity of the multiplication matrices is used.
%
%The {\sf BMS} algorithm, as the multi-dimensional generalization of the Berlekamp--Massey algorithm, worths further study. Related problems include, for example, the design of a fast variant with a better complexity, the probability for $I = I(E)$, and the new strategy when it fails, etc.

~\newline
{\bf Acknowledgements}
~\newline

The authors would like to thank Daniel Augot for his very helpful comments on the incremental Wiedemann algorithm and literatures about the BMS algorithm. This work is supported by the EXACTA grant of the French National Research Agency (ANR-09-BLAN-0371-01) and the National Science Foundation of China (NSFC 60911130369), the HPAC grant of the French National Research Agency, and the ECCA project by the Sino-French Laboratory for Computer Science, Automation and Applied Mathematics.
%\todo{Number}

\bibliographystyle{plain}

\begin{thebibliography}{10}

\bibitem{AJL2001g}
E.~Aguirre, A.S. Jarrah, and R.~Laubenbacher.
\newblock Generic ideals and {Moreno-Soc{\'\i}as} conjucture.
\newblock In {\em Proceedings of ISSAC 2001}, pages 21--23. ACM Press, 2001.

\bibitem{B2004e}
M.~Bardet.
\newblock {\'E}tude des syst{\`e}mes alg{\'e}briques surd{\'e}termin{\'e}s.
  {A}pplications aux codes correcteurs et {\`a} la cryptographie.
\newblock {\em PhD thesis, Universit{\'e} Paris VI}, 2004.

\bibitem{BFS04}
M.~Bardet, J.-C. Faug\`{e}re, and B.~Salvy.
\newblock On the complexity of {Gr{\"o}bner} basis computation of semi-regular
  overdetermined algebraic equations.
\newblock In {\em International Conference on Polynomial System Solving -
  ICPSS}, pages 71--75, 2004.

\bibitem{BF03C}
A.~Basiri and J.-C. Faug{\`e}re.
\newblock {Changing the ordering of Gr{\"o}bner bases with LLL: case of two
  variables}.
\newblock In {\em Proceedings of ISSAC 2003}, pages 23--29. ACM Press, 2003.

\bibitem{B1994T}
E.~Becker, T.~Mora, M.G. Marinari, and C.~Traverso.
\newblock {The shape of the Shape Lemma}.
\newblock In {\em Proceedings of ISSAC 1994}, pages 129--133. ACM Press, 1994.

\bibitem{B93G}
Thomas Becker, Volker Weispfenning, and Heinz Kredel.
\newblock {\em {Gr{\"o}bner Bases: a Computational Approach to Commutative
  Algebra}}.
\newblock Graduate Texts in Mathematics. Springer, New York, 1993.

\bibitem{BO2006C}
M.~Bras-Amor{\'o}s and M.E. O'Sullivan.
\newblock {The correction capability of the Berlekamp--Massey--Sakata algorithm
  with majority voting}.
\newblock {\em Applicable Algebra in Engineering, Communication and Computing},
  17(5):315--335, 2006.

\bibitem{Brent1980}
Richard~P. Brent, Fred~G. Gustavson, and David Y.~Y. Yun.
\newblock {Fast solution of Toeplitz systems of equations and computation of
  Pad\'e approximants}.
\newblock {\em Journal of Algorithms}, 1(3):259--295, 1980.

\bibitem{B85G}
B.~Buchberger.
\newblock Gr\"{o}bner bases: An algorithmic method in polynomial ideal theory.
\newblock In {\em Multidimensional Systems Theory}, pages 184--232. Reidel,
  Dordrecht, 1985.

\bibitem{BPW06}
J.~Buchmann, A.~Pyshkin, and R.-P. Weinmann.
\newblock A zero-dimensional {Gr{\"o}bner} basis for {AES-128}.
\newblock In Matthew Robshaw, editor, {\em Fast Software Encryption}, volume
  4047 of {\em LNCS}, pages 78--88. Springer, Berlin/Heidelberg, 2006.

\bibitem{C2007g}
M.~Cimpoeas.
\newblock Generic initial ideal for complete intersections of embedding
  dimension three with strong {Lefschetz} property.
\newblock {\em Bulletin Math\'{e}matique de la Soci\'{e}t\'{e} des Sciences
  Math\'{e}matiques de Roumanie. Nouvelle S\'{e}rie}, 50 (98)(1):33--66, 2007.

\bibitem{CKM97C}
S.~Collart, M.~Kalkbrener, and D.~Mall.
\newblock Converting bases with the {G}r\"{o}bner walk.
\newblock {\em Journal of Symbolic Computation}, 24(3--4):465--469, 1997.

\bibitem{CLO1998U}
D.A. Cox, J.B. Little, and D.~O'Shea.
\newblock {\em {Using Algebraic Geometry}}.
\newblock Springer Verlag, 1998.

\bibitem{DJMS08}
X.~Dahan, X.~Jin, M.~{Moreno Maza}, and E.~Schost.
\newblock Change of order for regular chains in positive dimension.
\newblock {\em Theoretical Computer Science}, 392:37--65, 2008.

\bibitem{FM11}
J.-C. Faug\`{e}re and C.~Mou.
\newblock Fast algorithm for change of ordering of zero-dimensional
  {Gr\"{o}bner} bases with sparse multiplication matrices.
\newblock In {\em Proceedings of ISSAC 2011}, pages 115--122. ACM Press, 2011.

\bibitem{FSS10}
J.-C. Faug\`{e}re, M.~{Safey El Din}, and P.-J. Spaenlehauer.
\newblock Computing loci of rank defects of linear matrices using {Gr\"{o}bner}
  bases and applications to cryptology.
\newblock In {\em Proceedings of ISSAC 2010}, pages 257--264. ACM Press, 2010.

\bibitem{F1999A}
Jean-Charles Faug\`{e}re.
\newblock A new efficient algorithm for computing {Gr\"{o}bner} bases
  (${F_4}$).
\newblock {\em Journal of Pure and Applied Algebra}, 139(1--3):61--88, 1999.

\bibitem{F2002A}
Jean-Charles Faug\`{e}re.
\newblock A new efficient algorithm for computing {Gr\"{o}bner} bases without
  reduction to zero (${F_5}$).
\newblock In {\em Proceedings of ISSAC 2002}, pages 75--83. ACM Press, 2002.

\bibitem{FGLM93E}
Jean-Charles Faug\`{e}re, P~Gianni, D~Lazard, and T~Mora.
\newblock Efficient computation of zero-dimensional {Gr\"{o}bner} bases by
  change of ordering.
\newblock {\em Journal of Symbolic Computation}, 16(4):329--344, 1993.

\bibitem{FR1993D}
G.L. Feng and T.R.N. Rao.
\newblock {Decoding algebraic-geometric codes up to the designed minimum
  distance}.
\newblock {\em IEEE Transactions on Information Theory}, 39(1):37--45, 1993.

\bibitem{H1998a}
T.~H{\o}holdt, J.H. van Lint, and R.~Pellikaan.
\newblock {\em Algebraic Geometry Codes}.
\newblock Handbook of Coding Theory. Elsevier, Amsterdam, 1998.

\bibitem{Jo1989}
Edmund Jonckheere and Chingwo Ma.
\newblock {A simple Hankel interpretation of the Berlekamp--Massey algorithm}.
\newblock {\em Linear Algebra and its Applications}, 125:65--76, 1989.

\bibitem{L83G}
D.~Lazard.
\newblock {Gr{\"o}bner bases, Gaussian elimination and resolution of systems of
  algebraic equations}.
\newblock In {\em Computer Algebra, EUROCAL' 83}, pages 146--156. Springer,
  1983.

\bibitem{L1992S}
D.~Lazard.
\newblock {Solving zero-dimensional algebraic systems}.
\newblock {\em Journal of symbolic computation}, 13(2):117--131, 1992.

\bibitem{LY1997O}
P.~Loustaunau and E.V. York.
\newblock {On the decoding of cyclic codes using Gr{\"o}bner bases}.
\newblock {\em Applicable Algebra in Engineering, Communication and Computing},
  8(6):469--483, 1997.

\bibitem{M1991a}
G.~Moreno-Soc{\'\i}as.
\newblock Autour de la fonction de {Hilbert-Samuel} (escaliers d'id\'{e}aux
  polynomiaux).
\newblock {\em PhD thesis, Ecole Polytechnique}, 1991.

\bibitem{M2003d}
G.~Moreno-Soc{\'\i}as.
\newblock {Degrevlex Gr{\"o}bner bases of generic complete intersections}.
\newblock {\em Journal of Pure and Applied Algebra}, 180(3):263--283, 2003.

\bibitem{P2010g}
K.~Pardue.
\newblock Generic sequences of polynomials.
\newblock {\em Journal of Algebra}, 324(4):579--590, 2010.

\bibitem{PC06}
C.~Pascal and E.~Schost.
\newblock Change of order for bivariate triangular sets.
\newblock In {\em Proceedings of ISSAC 2006}, pages 277--284. ACM Press, 2006.

\bibitem{R1999S}
F.~Rouillier.
\newblock Solving zero-dimensional systems through the rational univariate
  representation.
\newblock {\em Applicable Algebra in Engineering, Communication and Computing},
  9(5):433--461, 1999.

\bibitem{SH2002A}
K.~Saints and C.~Heegard.
\newblock {Algebraic-geometric codes and multidimensional cyclic codes: a
  unified theory and algorithms for decoding using Gr{\"o}bner bases}.
\newblock {\em IEEE Transactions on Information Theory}, 41(6):1733--1751,
  2002.

\bibitem{S88F}
S.~Sakata.
\newblock {Finding a minimal set of linear recurring relations capable of
  generating a given finite two-dimensional array}.
\newblock {\em Journal of Symbolic Computation}, 5(3):321--337, 1988.

\bibitem{S90E}
S.~Sakata.
\newblock {Extension of the Berlekamp--Massey algorithm to $N$ dimensions}.
\newblock {\em Information and Computation}, 84(2):207--239, 1990.

\bibitem{VG2003M}
J.~Von Zur~Gathen and J.~Gerhard.
\newblock {\em {Modern Computer Algebra}}.
\newblock Cambridge University Press, 2003.

\bibitem{W2001E}
D.~Wang.
\newblock {\em {Elimination Methods}}.
\newblock Springer Verlag, 2001.

\bibitem{W86S}
D.~Wiedemann.
\newblock {Solving sparse linear equations over finite fields}.
\newblock {\em IEEE Transactions on Information Theory}, 32(1):54--62, 1986.

\end{thebibliography}

\end{document}